%% file: main.tex
\documentclass{lmcs}
\pdfoutput=1

\usepackage{lastpage}
\lmcsdoi{17}{1}{21}
\lmcsheading{}{\pageref{LastPage}}{}{}%
{Dec.~23,~2019}{Mar.~18,~2021}{}

\usepackage[utf8]{inputenc}


\include{preamble}

\theoremstyle{plain}


\newcommand{\mathindex}[2]{\ensuremath{{#1}_{#2}}}

\begin{document}

\title[Stubborn Set Reduction for Two-Player Reachability Games]
{Stubborn Set Reduction for \texorpdfstring{\\}{} Two-Player Reachability Games}

\author[F.M.~B{\o}nneland]{Frederik~Meyer~B{\o}nneland}
\author[P.G.~Jensen]{Peter~Gj{\o}l~Jensen}
\author[K.G.~Larsen]{Kim~Guldstrand~Larsen}
\author[M.~Mu\~{n}iz]{Marco~Mu\~{n}iz}
\author[J.~Srba]{\texorpdfstring{Ji\v{r}\'{\i}}{Jir\'{\i}}~Srba\texorpdfstring{\vspace{-1cm}}{}}

\address{Department of Computer Science, Aalborg University, Denmark}
\email{\{frederikb,pgj,kgl,muniz,srba\}@cs.aau.dk}

\input{sections/abstract}

\maketitle

\input{sections/introduction}
\input{sections/preliminaries}
\input{sections/reduction}
\input{sections/petrigames}
\input{sections/implementation}
\input{sections/conclusion}

\bibliographystyle{alpha}
\bibliography{references}{}
\end{document}

%% file: preamble.tex
\usepackage{amssymb}
\usepackage[utf8]{inputenc}
\usepackage{syntax}
\usepackage{graphicx}
\usepackage{amsmath}
\usepackage{url}
\usepackage{xspace}
\usepackage{tikz}
\usetikzlibrary{arrows,shapes,automata,backgrounds,petri}
\usetikzlibrary{positioning}
\usepackage{xcolor}
\usepackage[ruled,vlined,linesnumbered]{algorithm2e}

\SetCommentSty{mycommfont}
\SetKwInOut{Parameter}{parameter}
\SetKwRepeat{Do}{do}{while}
\usepackage{centernot}
\usepackage{multirow}
\usepackage{pbox}
\usepackage{float}
\usepackage[T1]{fontenc}
\usepackage{makecell}

\newcommand{\trailingspace}[1]{#1\xspace} 

\SetKwProg{Fn}{Function}{}{}
\SetKwInOut{Input}{input}
\SetKwInOut{Output}{output}
\DontPrintSemicolon%

\usepackage{array}
\newcolumntype{L}[1]{>{\raggedright\let\newline\\\arraybackslash\hspace{0pt}}m{#1}}
\newcolumntype{C}[1]{>{\centering\let\newline\\\arraybackslash\hspace{0pt}}m{#1}}
\newcolumntype{R}[1]{>{\raggedleft\let\newline\\\arraybackslash\hspace{0pt}}m{#1}}

\newcommand{\PNG}{\trailingspace{Petri net game}} 

\newcommand{\postset}[1]{#1^{\bullet}}
\newcommand{\preset}[1]{{}^{\bullet}#1}
\newcommand{\inhibpostset}[1]{#1^{\circ}}
\newcommand{\inhibpreset}[1]{{}^{\circ}#1}
\newcommand{\preincr}[1]{{}^{+}#1}
\newcommand{\predecr}[1]{{}^{-}#1}
\newcommand{\postdecr}[1]{{#1}^{-}}
\newcommand{\postincr}[1]{{#1}^{+}}

\newcommand{\goal}{\ensuremath{\mathit{Goal}}\xspace} 

\newcommand{\safe}[1]{\ensuremath{\mathit{safe}(#1)}\xspace}
\newcommand{\interest}[2]{\ensuremath{A_{#1}(#2)}\xspace}

\newcommand{\stratnext}[2]{\ensuremath{\mathit{next}_{#1}(#2)}\xspace}

\usepackage{xparse}

\NewDocumentCommand{\rtrans}{oo}
{\IfNoValueTF{#1}
        {\ensuremath{\xrightarrow{}}\xspace}
        {\IfNoValueTF{#2}
            {\ensuremath{\xrightarrow{#1}}\xspace}
            {\ensuremath{\xrightarrow[#2]{#1}}\xspace}}}

\NewDocumentCommand{\rntrans}{oo}
{\IfNoValueTF{#1}
        {\ensuremath{\centernot{\xrightarrow{}}}\xspace}
        {\IfNoValueTF{#2}
            {\ensuremath{\centernot{\xrightarrow{#1}}}\xspace}
            {\ensuremath{\centernot{\xrightarrow[#2]{#1}}}\xspace}}}

\NewDocumentCommand{\lntrans}{oo}
{\IfNoValueTF{#1}
        {\ensuremath{\centernot{\leftarrow{}}}\xspace}
        {\IfNoValueTF{#2}
            {\ensuremath{\centernot{\leftarrow{#1}}}\xspace}
            {\ensuremath{\centernot{\leftarrow[#2]{#1}}}\xspace}}}

\NewDocumentCommand{\interesting}{oo}
{\IfNoValueTF{#1}
        {\ensuremath{A}\xspace}
        {\IfNoValueTF{#2}
            {\ensuremath{A_{#1}}\xspace}
            {\ensuremath{A_{#1}(#2)}\xspace}}}

\NewDocumentCommand{\outcome}{oo}
{\IfNoValueTF{#1}
        {\ensuremath{\mathit{Outcome}}\xspace}
        {\IfNoValueTF{#2}
            {\ensuremath{\mathit{Outcome}(#1)}\xspace}
            {\ensuremath{\mathit{Outcome}_(#1,#2)}\xspace}}}

\NewDocumentCommand{\incr}{oo}
{\IfNoValueTF{#1}
        {\ensuremath{\mathit{incr}}\xspace}
        {\IfNoValueTF{#2}
            {\ensuremath{\mathit{incr}_{#1}}\xspace}
            {\ensuremath{\mathit{incr}_{#1}(#2)}\xspace}}}

\NewDocumentCommand{\decr}{oo}
{\IfNoValueTF{#1}
        {\ensuremath{\mathit{decr}}\xspace}
        {\IfNoValueTF{#2}
            {\ensuremath{\mathit{decr}_{#1}}\xspace}
            {\ensuremath{\mathit{decr}_{#1}(#2)}\xspace}}}

\NewDocumentCommand{\eval}{oo}
{\IfNoValueTF{#1}
        {\ensuremath{\mathit{eval}}\xspace}
        {\IfNoValueTF{#2}
            {\ensuremath{\mathit{eval}_{#1}}\xspace}
            {\ensuremath{\mathit{eval}_{#1}(#2)}\xspace}}}

\NewDocumentCommand{\simplify}{o}
    {\IfNoValueTF{#1}
        {\ensuremath{\mathit{simplify}}\xspace}
        {\ensuremath{\mathit{simplify}(#1)}\xspace}}

\newcommand{\pathlen}{\ensuremath{\ell}\xspace}
\newcommand{\paths}{\ensuremath{\Pi}\xspace}
\newcommand{\maxpaths}{\ensuremath{\Pi^{\mathit{max}}}\xspace} 

\newcommand{\places}{\ensuremath{P}\xspace} 
\newcommand{\transitions}{\ensuremath{T}\xspace} 
\newcommand{\weights}{\ensuremath{W}\xspace} 
\newcommand{\markings}{\ensuremath{\mathcal{M}(N)}\xspace} 
\newcommand{\inhib}{\ensuremath{I}\xspace} 
\newcommand{\petrituple}{\ensuremath{(\places,\transitions_1, \transitions_2, \weights,\inhib)}\xspace} 

\NewDocumentCommand{\reduction}{o} 
    {\IfNoValueTF{#1}
        {\ensuremath{St}\xspace}
        {\ensuremath{St(#1)}\xspace}}

\NewDocumentCommand{\comp}{o} 
    {\IfNoValueTF{#1}
        {\ensuremath{\mathit{comp}}\xspace}
        {\ensuremath{\mathit{comp}(#1)}\xspace}}

\newcommand{\lts}{\ensuremath{G}\xspace} 
\newcommand{\ltsstates}{\ensuremath{\mathcal{S}}\xspace} 
\newcommand{\ltslabels}{\ensuremath{A}\xspace} 
\newcommand{\ltsedges}{\ensuremath{\mathcal{\rightarrow}}\xspace} 

\newcommand{\deadlock}{\ensuremath{\mathit{deadlock}}\xspace}
\newcommand{\true}{\ensuremath{\mathit{true}}\xspace}
\newcommand{\false}{\ensuremath{\mathit{false}}\xspace}

\newcommand{\en}{\ensuremath{en}\xspace} 

\makeatletter
\def\namedlabel#1#2{\begingroup
   \def\@currentlabel{#2}%
   \label{#1}\endgroup
}
\makeatother

\makeatletter
\newcommand*{\da@rightarrow}{\mathchar"0\hexnumber@\symAMSa 4B } 
\newcommand*{\da@leftarrow}{\mathchar"0\hexnumber@\symAMSa 4C } 
\newcommand*{\xdashrightarrow}[2][]{%
  \mathrel{%
    \mathpalette{\da@xarrow{#1}{#2}{}\da@rightarrow{\,}{}}{}%
  }%
}
\newcommand{\xdashleftarrow}[2][]{%
  \mathrel{%
    \mathpalette{\da@xarrow{#1}{#2}\da@leftarrow{}{}{\,}}{}%
  }%
}
\newcommand*{\da@xarrow}[7]{%
  \sbox0{$\ifx#7\scriptstyle\scriptscriptstyle\else\scriptstyle\fi#5#1#6\m@th$}%
  \sbox2{$\ifx#7\scriptstyle\scriptscriptstyle\else\scriptstyle\fi#5#2#6\m@th$}%
  \sbox4{$#7\dabar@\m@th$}%
  \dimen@=\wd0 %
  \ifdim\wd2 >\dimen@ 
    \dimen@=\wd2 %
  \fi
  \count@=2 %
  \def\da@bars{\dabar@\dabar@}%
  \@whiledim\count@\wd4<\dimen@\do{
    \advance\count@\@ne 
    \expandafter\def\expandafter\da@bars\expandafter{%
      \da@bars 
      \dabar@ 
    }%
  }%
  \mathrel{#3}%
  \mathrel{%
    \mathop{\da@bars}\limits 
    \ifx\\#1\\%
    \else
      _{\copy0}%
    \fi
    \ifx\\#2\\%
    \else
      ^{\copy2}%
    \fi
  }%
  \mathrel{#4}%
}
\makeatother

\newcommand{\ubound}{\mathit{ub}}
\newcommand{\lbound}{\mathit{lb}}

\usepackage{booktabs}

%% file: sections/abstract.tex
\begin{abstract}
Partial order reductions have been successfully applied
to model checking of concurrent systems and practical applications
of the technique show nontrivial reduction in the size of the
explored state space.
We present a theory of partial order reduction based on stubborn sets
in the game-theoretical
setting of 2-player games with reachability objectives.
Our stubborn reduction allows us to prune the interleaving behaviour of
both players in the game, and we formally prove its correctness on the
class of games played on general labelled transition systems.
We then instantiate the framework to the class of
weighted Petri net games with inhibitor arcs and provide its efficient
implementation in the model checker TAPAAL\@. Finally, we evaluate our stubborn
reduction on several case studies and demonstrate its efficiency.
\end{abstract}

%% file: sections/introduction.tex
\section{Introduction}
The state space explosion problem is the main obstacle for model checking
of concurrent systems. Even simple processes running in parallel
can produce an exponentially large number of interleavings,
making full state space search practically intractable.
A family of methods for taming this problem is that of
partial order reductions~\cite{G:96,peled1993stubborn,V:90} by exploiting the commutativity of independent concurrent processes.
Variants of partial order reductions include persistent
sets~\cite{G:96,godefroid1990using,godefroid1993using},
ample sets~\cite{peled1993stubborn,peled1996combining,peled1998ten},
and stubborn sets~\cite{V:90,valmari1992attack,valmari1993partial,valmari2017stubborn}.

As our main contribution,
we generalise the theory of the stubborn set variant of partial order reductions into
the setting of 2-player games.
We exploit the observation that either of the two players often is
left with no actions to propose, leaving the opponent to independently
dictate the behavior of the system for a limited, consecutive sequence of actions.
In such cases we may
apply the classical stubborn set reductions in order to
reduce the number of interleavings of independent actions.
To preserve the winning strategies of both players,
 a number of conditions of the reduction has to be satisfied.
 We define the notion of a
\emph{stable} stubborn set reduction by a set of sufficient conditions that guarantee the preservation
of winning strategies for both players.
Furthermore, we formally prove the correctness of stable reductions
in the setting of general game labelled transition systems,
and instantiate our framework to weighted
Petri net games with inhibitor arcs.
We propose approximate syntax-driven conditions of a stable Petri net game reduction
satisfying the sufficient conditions for our stable reductions
and demonstrate their applicability in an efficient, open source implementation in the model
checker TAPAAL~\cite{david2012tapaal}.
Our implementation is based on
dependency graphs, following the approach from~\cite{DEFJJJKLNOPS:FI:18,jensen2016real},
and we demonstrate on  several case studies  that the computation of the stubborn sets
only has a minor overhead while having the potential of achieving exponential
reduction both in the running time as well as in
the number of searched configurations.
To the best of our knowledge, this is the first efficient implementation of a 2-player game partial
order reduction technique for Petri net games.

\emph{Related Work.}
Partial order reductions in the non-game setting
for linear time properties have previously been studied~\cite{laarman2014real,lehmann2012stubborn,peled1993stubborn,valmari1992attack} which lends itself towards the safety or liveness properties we wish to preserve for winning states.
Originally Peled and Valmari presented partial order reductions for general stuttering-insensitive LTL~\cite{peled1993stubborn,valmari1992attack} and
Lehmann et al.\ subsequently studied stubborn sets applied to a subset of LTL properties, called simple linear time properties, allowing them to utilise a relaxed set of conditions compared to those for general LTL preservation~\cite{lehmann2012stubborn}.

The extension of partial order reductions to game-oriented formalisms and verification tasks has not yet received much attention in the literature.
In~\cite{jamroga2018towards} partial order reductions for LTL without
the next operator are adapted to a subset of alternating-time temporal logic
and applied to multi-agent systems.
The authors consider games with imperfect information, however, they also show that their technique is inapplicable for perfect information games.
In our work, we assume an antagonistic environment and focus on preserving
the existence of winning strategies with perfect information, reducing the state space, and improving existing controller synthesis algorithms.
Partial order reduction for the problem of checking bisimulation equivalence between two labelled transition systems is presented in~\cite{valmari1997set,huhn1998partial,gerth1999partial}.
Our partial order reduction is applied directly to a labelled transition system while theirs are applied to the bisimulation game graph.
While the setting is distinctly different, our approach is more general as we allow for mixed states
and allow for reduction in both controllable as well
as environmental states.
Moreover, we provide an implementation of the
on-the-fly strategy synthesis algorithm and argue by a number of case
studies for its practical applicability.

The work on partial order reductions for weak modal $\mu$-calculus and CTL
(see e.g.~\cite{RS:CONCUR:97,561357})
allows us in principle to encode the game semantics as a part of a
$\mu$-calculus formula.
Recently, partial order reduction techniques for parity games have been proposed by Neele et al.~\cite{neele2020partial}, which allows for model checking the full modal $\mu$-calculus.
However, the use of more general partial order reduction methods may waste reduction potential, as the more general methods usually generate larger stubborn sets to preserve properties that are not required in the 2-player game setting.

Complexity and decidability results for control synthesis in Petri net games are not encouraging.
The control synthesis problem is for many instances of Petri net formalisms undecidable~\cite{alechina2016complexity,berard2012concurrent}, including those that allow for inhibition~\cite{berard2012concurrent} which we utilise to model our case studies.
If the problem is decidable for a given instance of a Petri net formalism
(like e.g.\ for bounded nets) then it is usually of high computational complexity.
In fact, most questions about the behaviour of bounded Petri nets are
at least PSPACE-hard~\cite{esparza1998decidability}.
We opt to use efficient overapproximation algorithms using both syntactic and local state information to generate stable stubborn sets.

The work presented in this article is an extended version with full proofs of
our conference paper~\cite{boenneland2019partial}. The stubborn set
conditions presented in~\cite{boenneland2019partial} were insufficient
in order to guarantee the preservation of reachability while condition {\bf C}
from the conference paper was found to be redundant.  These issues are fixed,
and in the present article we add an additional visibility condition on
player 2 actions and we elaborate on its syntax-based
algorithmic overapproximation for the Petri net
games. The implementation is accordingly fixed and the efficiency of the
method is still confirmed on an extended set of case studies compared
to~\cite{boenneland2019partial}.

%% file: sections/preliminaries.tex
\section{Preliminaries}

We shall first introduce the basic notation and definitions.

\begin{defi}[Game Labelled Transition System]
	A (deterministic) Game Labelled Transition System (GLTS) is a tuple $\lts = (\ltsstates, \ltslabels_1, \ltslabels_2, \ltsedges, \goal)$ where
\begin{itemize}
\item $\ltsstates$ is a set of states,
\item $\ltslabels_1$ is a finite set of actions for player~$1$ (the controller),
\item $\ltslabels_2$ is a finite set of actions for player~$2$ (the environment)
where $\ltslabels_1 \cap \ltslabels_2 = \emptyset$
and $\ltslabels = \ltslabels_1 \cup \ltslabels_2$,
\item $\ltsedges \subseteq \ltsstates \times \ltslabels \times \ltsstates$ is a
transition relation such that if $(s,a,s') \in \ltsedges$ and $(s,a,s'') \in \ltsedges$ then $s' = s''$, and
\item $\goal \subseteq \ltsstates$ is a set of goal states.
\end{itemize}
\end{defi}

\noindent
Let $\lts = (\ltsstates, \ltslabels_1, \ltslabels_2, \ltsedges, \goal)$ be a fixed GLTS for the remainder of the section.
Whenever $(s,a,s') \in \ltsedges$ we write $s \rtrans[a] s'$ and say that $a$ is enabled in $s$ and can be \emph{executed} in $s$ yielding $s'$.
Otherwise we say that $a$ is \emph{disabled} in $s$.
The set of \emph{enabled} player~$i$ actions where $i \in \{1,2\}$ in a state $s \in \ltsstates$ is given by $\en_i(s) = \{ a \in \ltslabels_i \mid \exists s' \in \ltsstates.\ s \rtrans[a] s'\}$.
The set of all enabled actions is given by $\en(s) = \en_1(s) \cup \en_2(s)$.
For a state $s \in \ltsstates$ where $\en(s) \neq \emptyset$ if $\en_2(s) = \emptyset$ then we call $s$ a player~$1$ state, if $\en_1(s) = \emptyset$ then we call $s$ a player~$2$ state, and otherwise we call it a \emph{mixed} state.
If $\en(s) = \emptyset$ then we call $s$ a \emph{deadlock} state.
The GLTS $\lts$ is called \emph{non-mixed} if all states are either player~$1$, player~$2$, or deadlock states.

For a sequence of actions $w = a_1 a_2 \cdots a_n \in \ltslabels^*$ we write $s \rtrans[w] s'$ if $s \rtrans[a_1] s_1 \rtrans[a_2] \cdots \rtrans[a_n] s'$ and say it is \emph{executable}.
If $w \in \ltslabels^{\omega}$, i.e.\ if it is infinite, then we write $s \rtrans[w]$.
Actions that are a part of $w$ are said to occur in $w$.
A sequence of states induced by $w \in \ltslabels^* \cup \ltslabels^{\omega}$ is called a \emph{run} and is written as $\pi = s_0 s_1\cdots$.
We use $\paths_{\lts}(s)$ to denote the set of all runs starting from a state $s \in \ltsstates$ in GLTS $\lts$, s.t.\ for all $s_0 s_1 \cdots \in \paths_G(s)$ we have $s_0 = s$, and $\paths_{\lts} = \bigcup_{s \in \ltsstates} \paths_{\lts}(s)$ as the set of all runs.
The number of actions in a run $\pi$ is given by the function $\pathlen : \paths_G \to \mathbb{N}^0 \cup \{ \infty \}$ s.t.\ for a run $\pi = s_0 \cdots s_n$ we have $\pathlen(\pi) = n$ if $\pi$ is finite and otherwise $\pathlen(\pi) = \infty$.
A position in a run $\pi = s_0 s_1 \ldots \in \paths_{\lts}(s)$ is a natural number $i \in \mathbb{N}^0$ that refers to the state $s_i$ and is written as $\pi_i$.
A position $i$ can range from $0$ to $\pathlen(\pi)$ s.t.\ if $\pi$ is infinite then $i \in \mathbb{N}^0$ and otherwise $0 \leq i \leq \pathlen(\pi)$.
Let $\maxpaths_{\lts}(s)$ be the set of all maximal runs starting from $s$, defined as $\maxpaths_{\lts}(s) = \{ \pi \in \paths_{\lts}(s) \mid \pathlen(\pi) \neq \infty \implies \en(\pi_{\pathlen(\pi)}) = \emptyset \}$.
We omit the GLTS $\lts$ from the subscript of run sets if it is clear from the context.

A reduced game is defined by a function called a reduction.
\begin{defi}[Reduction]
	Let $\lts = (\ltsstates, \ltslabels_1, \ltslabels_2, \ltsedges, \goal)$ be a GLTS\@.
	A \emph{reduction} is a function $\reduction : \ltsstates \to 2^\ltslabels$.
\end{defi}

\begin{defi}[Reduced Game]
	Let $\lts = (\ltsstates, \ltslabels_1, \ltslabels_2, \ltsedges, \goal)$ be a GLTS and $\reduction$ be a reduction.
	The \emph{reduced game} of $\lts$ by the reduction $\reduction$ is given by $\lts_{\reduction} = (\ltsstates, \ltslabels_1, \ltslabels_2, \rtrans[][\reduction], \goal)$ where $s \rtrans[a][\reduction] s'$ iff $s \rtrans[a] s'$ and $a \in \reduction(s)$.
\end{defi}

The set of actions $\reduction(s)$ is the \emph{stubborn set} of $s$ with the reduction $\reduction$.
The set of non-stubborn actions for $s$ is defined as $\overline{\reduction(s)} = \ltslabels \setminus \reduction(s)$.

A (memoryless) strategy is a function that proposes the next action player~$1$
wants to execute.
\begin{defi}[Strategy]
	Let $\lts = (\ltsstates, \ltslabels_1, \ltslabels_2, \ltsedges, \goal)$ be a GLTS\@.
	A \emph{strategy} is a function $\sigma: \ltsstates \to \ltslabels_1 \cup \{ \bot \}$ where for all $s \in \ltsstates$ we have that if $\en_1(s) \neq \emptyset$ then $\sigma(s) \in \en_1(s)$ else $\sigma(s) = \bot$.
\end{defi}
The intuition is that in order
to ensure progress, player~$1$ always has to propose an action
if she has an enabled action.
Let $\sigma$ be a fixed strategy for the remainder of the section.
We define a function $\stratnext{\sigma}{s}$ that returns the set of actions considered at $s \in \ltsstates$ under $\sigma$ as:
\[
\stratnext{\sigma}{s} =
\begin{cases}
	\en_2(s) \cup \sigma(s)\ &\text{if } \sigma(s) \neq \bot \\
	\en_2(s)\ &\text{otherwise.}
\end{cases}
\]
Let $\maxpaths_{\sigma}(s) \subseteq \maxpaths(s)$ be the set of maximal runs subject to $\sigma$ starting at $s \in \ltsstates$, defined as:
\[\maxpaths_{\sigma}(s) = \{\pi \in \maxpaths(s) \mid \forall i \in \{1,\ldots,\pathlen(\pi)\}.\ \exists a \in \stratnext{\sigma}{\pi_{i-1}}.\ \pi_{i-1} \rtrans[a] \pi_{i}\} \ .
\]

\begin{defi}[Winning Strategy]
	Let $\lts = (\ltsstates, \ltslabels_1, \ltslabels_2, \ltsedges, \goal)$ be a GLTS and $s \in \ltsstates$ be a state.
	A strategy $\sigma$ is a \emph{winning strategy} for player~$1$ at $s$ in $\lts$ iff for all $\pi \in \maxpaths_{\sigma}(s)$ there exists a position $i$ s.t.\ $\pi_i \in \goal$.
A state $s$ is called \emph{winning} if there is a winning strategy for player~$1$ at $s$.
\end{defi}

If a state is winning for player~$1$ in $\lts$ then no matter what
action sequence the environment chooses, eventually a goal state is reached.
Furthermore, for a given winning strategy $\sigma$ at $s$ in $\lts$, there is a finite number $n \in \mathbb{N}^0$ such that we always reach a goal state with at most $n$ action firings, which we prove in Lemma~\ref{lemma:depth}.
We call this minimum
number the \emph{strategy depth} of $\sigma$.
\begin{defi}[Strategy Depth]
	Let $\lts = (\ltsstates, \ltslabels_1, \ltslabels_2, \ltsedges, \goal)$ be a GLTS, $s \in \ltsstates$ a winning state for player~$1$ in $\lts$
and $\sigma$ a winning strategy at $s$ in $\lts$.
	Then $n \in \mathbb{N}^0$ is the \emph{depth} of $\sigma$ at $s$ in $\lts$ if:
	\begin{itemize}
		\item for all $\pi \in \maxpaths_{\sigma}(s)$ there exists $0 \leq i \leq n$ s.t.\ $\pi_i \in \goal$, and
		\item there exists $\pi' \in \maxpaths_{\sigma}(s)$ s.t.\ $\pi'_n \in \goal$ and for all $0 \leq j < n$ we have $\pi'_j \notin \goal$.
	\end{itemize}
\end{defi}

\begin{lem}\label{lemma:depth}
	Let $\lts = (\ltsstates, \ltslabels_1, \ltslabels_2, \ltsedges, \goal)$ be a GLTS, $s \in \ltsstates$ a winning state for player~$1$ in $\lts$, and $\sigma$ a winning strategy at $s$ in $\lts$. Then
	\begin{enumerate}
		\item\label{stratA} there exists $n \in \mathbb{N}$ that is the depth of $\sigma$ at $s$ in $\lts$, and
		\item\label{stratB} if $s \notin \goal$ then for all $a \in \stratnext{\sigma}{s}$
where $s \rtrans[a] s'$, the depth of  $\sigma$ at $s'$ in $\lts$ is
$m$ such that $0 \leq m < n$.
	\end{enumerate}
\end{lem}
\begin{proof}
	(\ref{stratA}):
	Due to $\ltslabels_1$ and $\ltslabels_2$ being finite and any $\lts$ being deterministic, we know that
every state $s \in \ltsstates$ is finitely branching.
	Since $s$ is a winning state for player~$1$ in $\lts$,
we get that every run leads to a goal state in a finite number of actions.
	Therefore, due to K\"onig's lemma, the tree induced by all runs starting from $s$, with the leafs being the first occurring goal states, is a finite tree and
hence such $n$ exists.

	(\ref{stratB}):
	Let $n$ be the depth of $\sigma$ at $s$ in $\lts$
and let $s \rtrans[a] s'$ such that $a \in \stratnext{\sigma}{s}$.
By contradiction let us assume that
the depth of $\sigma$ at $s'$ is larger than or equal
to $n$. However, this implies the existence of a run $\pi$
from $s'$ that contains $n$ or more non-goal states before reaching
the goal. The run $s \pi$ now contradicts that the depth of $s$
is $n$.
\end{proof}

A set of actions for a given state and a given set of goal states is called an
\emph{interesting set} if for any run leading to any goal state at least one action from the set of interesting actions has to be executed.
\begin{defi}[Interesting Actions]
	Let $\lts = (\ltsstates, \ltslabels_1, \ltslabels_2, \ltsedges, \goal)$ be a GLTS and $s \in \ltsstates$ a state.
	A set of actions $\interest{s}{\goal} \subseteq \ltslabels$ is called an \emph{interesting set} of actions for $s$ and $\goal$ if whenever $s \notin \goal$, $w = a_1 \cdots a_n \in \ltslabels^*$, $s \rtrans[w] s'$, and $s' \in \goal$ then there exists $i$, $1 \leq i \leq n$, such that $a_i \in \interest{s}{\goal}$.
\end{defi}

\begin{exa}\label{ex:interesting}
In Figure~\ref{fig:safe-interesting-example} we see an example of a GLTS $\lts = (\ltsstates, \ltslabels_1, \ltslabels_2, \ltsedges, \goal)$ where $\ltsstates = \{ s_1, s_2, s_3, s_4, s_5, s_6, s_7 \}$ are the states denoted by circles,
$\ltslabels_1 = \{a,b,c\}$ is the set of player~$1$ actions, $\ltslabels_2 = \{d\}$ is the set of player~$2$ actions, and $\ltsedges$ is denoted by the solid
(controllable) and dashed (uncontrollable) transitions between states,
labelled by the corresponding actions for player~$1$ and~$2$, respectively.
Let $\goal = \{ s_6 \}$. 
We now consider different proposals for a set of interesting actions for the state $s_1$.
The set $\{ b \}$ is an interesting set of actions in $s_1$ since the goal state $s_6$ cannot be reached without firing $b$ at least once.
Furthermore, the sets $\{ a \}$ and $\{ c \}$ are also sets of interesting actions for the state $s_1$.
\begin{figure}[t]
    	\centering
    	\begin{tikzpicture}[font=\scriptsize,xscale=2.2,yscale=1.3]
    	    \tikzstyle{state}=[inner sep=0pt,circle,draw=black,very thick,fill=white,minimum height=5mm, minimum width=5mm,font=\small]
        	\tikzstyle{empty}=[rectangle,draw=none,font=\small]
        	\tikzstyle{reducedstate}=[inner sep=0pt,circle,draw=black,fill=white,minimum height=5mm, minimum width=5mm,font=\small]
        	\tikzstyle{every label}=[black]
        	\tikzstyle{playerArc}=[->,>=stealth,very thick]
        	\tikzstyle{reducedplayerArc}=[->,>=stealth]
        	\tikzstyle{opponentArc}=[->,>=stealth,dashed,very thick]
        	\tikzstyle{reducedopponentArc}=[->,>=stealth,dashed]
		\begin{scope}
			\node [state] at (0,0) (s1) {$s_1$};
			\node [state] at (1.2,0) (s2) {$s_2$};
			\node [state] at (2.4,0) (s3) {$s_3$};
			\node [state] at (0,-1.2) (s4) {$s_4$};
			\node [state] at (1.2,-1.2) (s5) {$s_5$};
			\node [state,label=right:$\in \goal$] at (2.4,-1.2) (s6) {$s_6$};
			\node [state] at (1.2,-2.4) (s7) {$s_7$};

			\node [empty] at (0,-2.4) {$\safe{s_1} = \{ a \}$};
			\node [empty] at (0,-2.1) {$\interest{s_1}{\{s_6\}} = \{ a \}$};

			\draw[playerArc] (s1) -- (s2) node[midway,above]{$a$} {};
			\draw[playerArc] (s2) -- (s3) node[midway,above]{$b$} {};
			\draw[playerArc] (s1) -- (s4) node[midway,left]{$c$} {};
			\draw[playerArc] (s4) -- (s5) node[midway,above]{$a$} {};
			\draw[playerArc] (s5) -- (s6) node[midway,above]{$b$} {};
			\draw[playerArc] (s3) -- (s6) node[midway,right]{$c$} {};
			\draw[opponentArc] (s5) -- (s7) node[midway,right]{$d$} {};
		\end{scope}
    	\end{tikzpicture}
    \caption{Example of safe and interesting sets of actions for a state $\mathindex{s}{1}$}%
    \label{fig:safe-interesting-example}
\end{figure}
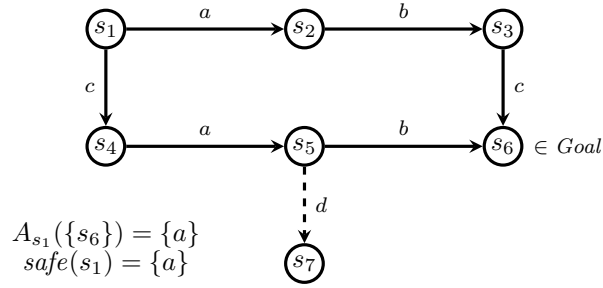
\end{exa}

Player~$1$ has to also consider her safe actions.
A player~$1$ action is \emph{safe} in a given player~$1$ state
if for any player~$1$ action sequence (excluding the safe action)
that does not enable any player~$2$ action, prefixing this sequence
with the safe action will (in case it is executable) also not enable any player~$2$ action.
\begin{defi}[Safe Action]
	Let $\lts = (\ltsstates, \ltslabels_1, \ltslabels_2, \ltsedges, \goal)$ be a GLTS and $s \in \ltsstates$ a state such that $\en_2(s) = \emptyset$.
	An action $a \in \en_1(s)$ is \emph{safe} in $s$ if whenever $w \in {(\ltslabels_1 \setminus \{a\})}^*$ with $s \rtrans[w] s'$ s.t.\ $\en_2(s') = \emptyset$ and $s \rtrans[aw] s''$ then $\en_2(s'') = \emptyset$.
	The set of all safe actions for $s$ is written as $\safe{s}$.
\end{defi}

\begin{exa}\label{ex:safe}
Consider again the GLTS in Figure~\ref{fig:safe-interesting-example}.
We reasoned in Example~\ref{ex:interesting} that the set $\{ b \}$ is an interesting set of actions in the state $s_1$.
However, $b$ is not a safe player~$1$ action in $s_1$ since by definition $b$ has to be enabled at $s_1$ to be safe.
The set of enabled actions in $s_1$ is $\en(s_1) = \{ a, c \}$, and between these two actions only $a$ is safe.
The action $c$ is not safe since we have $s_1 \rtrans[a] s_2$ and $\en_2(s_2) = \emptyset$ but $s_1 \rtrans[ca] s_5$ and $\en_2(s_5) \neq \emptyset$.
It is clear that $s_1$ is a winning state for player~$1$
and player~$1$ must initially play $a$ as playing $c$ will bring us
to the mixed state $s_5$ from which player~$1$ does not have winning strategy.
\end{exa}

%% file: sections/reduction.tex
\section{Stable Reduction}
In this section we introduce the notion of a \emph{stable} reduction
$\reduction$ that provides at each state $s$ the set of actions
$\reduction(s)$ that are sufficient to be explored so that the given
reachability property is preserved in the reduced game.
In the game setting, we have to guarantee the preservation of winning
strategies for both players in the game.
In what follows, we shall introduce a number of conditions (formulated
in general terms of game labelled transition systems) that
guarantee that a given reduction preserves winning strategies
and we shall call reductions satisfying these conditions \emph{stable}.

For the remainder of the section let $s \in S$ be a state and $\goal \subseteq \ltsstates$ be a set of goal states, and let $\interest{s}{\goal}$ be an arbitrary
but fixed set of interesting actions for $s$ and $\goal$.

\begin{defi}[Stable Reduction Conditions]
	A reduction $\reduction$ is called \emph{stable} if
$\reduction$ satisfies for every $s \in S$ Conditions~\ref{rule:stub-init},~\ref{rule:stub},~\ref{rule:reach},~\ref{rule:game-1},~\ref{rule:game-2},~\ref{rule:safe},~\ref{rule:visible} and~\ref{rule:dead}.
        \begin{itemize}[left=6mm]
          \item[\textbf{I}]\namedlabel{rule:stub-init}{\textbf{I}} If $\en_1(s) \neq \emptyset$ and $\en_2(s) \neq \emptyset$ then $\en(s) \subseteq \reduction(s)$.
          \item[\textbf{W}]\namedlabel{rule:stub}{\textbf{W}} For all $w \in \overline{\reduction(s)}^*$ and all $a \in \reduction(s)$ if $s \rtrans[wa] s'$ then $s \rtrans[aw] s'$.
          \item[\textbf{R}]\namedlabel{rule:reach}{\textbf{R}} $\interest{s}{\goal} \subseteq \reduction(s)$
          \item[\textbf{G1}]\namedlabel{rule:game-1}{\textbf{G1}} For all $w \in \overline{\reduction(s)}^*$ if $\en_2(s) = \emptyset$ and $s \rtrans[w] s'$ then $\en_2(s') = \emptyset$.
          \item[\textbf{G2}]\namedlabel{rule:game-2}{\textbf{G2}} For all $w \in \overline{\reduction(s)}^*$ if $\en_1(s) = \emptyset$ and $s \rtrans[w] s'$ then $\en_1(s') = \emptyset$.
          \item[\textbf{S}]\namedlabel{rule:safe}{\textbf{S}} $\en_1(s) \cap \reduction(s) \subseteq \safe{s}$ or $en_1(s) \subseteq \reduction(s)$
          \item[\textbf{V}]\namedlabel{rule:visible}{\textbf{V}} If there exists $w \in \ltslabels_2^*$ s.t.\ $s \rtrans[w] s'$ and $s' \in \goal$ then $\en_2(s) \subseteq \reduction(s)$.
          \item[\textbf{D}]\namedlabel{rule:dead}{\textbf{D}} If $\en_2(s) \neq \emptyset$ then there exists $a \in \en_2(s) \cap \reduction(s)$ s.t.\ for all $w \in \overline{\reduction(s)}^*$ where $s \rtrans[w] s'$ we have $a \in \en_2(s')$.
        \end{itemize}
\end{defi}

If $s$ is a mixed state then Condition~\ref{rule:stub-init} ensures that all enabled actions are included in the reduction. That is, we do not attempt to reduce the state space from this state.
Condition~\ref{rule:stub} states that we can swap the ordering of action
sequences such that performing stubborn actions first still ensures that
we can reach a given state (i.e.\ a stubborn action commutes with
any sequence of nonstubborn actions).
Condition~\ref{rule:reach} ensures that a goal state cannot be reached solely by exploring actions not in the stubborn set (i.e.\ we preserve the
reachability of goal states).
Conditions~\ref{rule:game-1} resp.~\ref{rule:game-2} ensure that
from any state belonging  to player~$1$ (resp.\ player~$2$), it is not possible
to reach any  player $2$ (resp.\ player $1$) state or a mixed state,  solely by exploring
only nonstubborn actions
(i.e.\ reachability of  mixed states and opposing player states are preserved in the reduction).
Condition~\ref{rule:safe} ensures that either all enabled stubborn player 1 actions
are also safe, or if this is not the case then
all enabled player $1$ actions are included in the stubborn set.
Condition~\ref{rule:visible} checks if it is possible to reach a goal state by firing exclusively player $2$ actions, and includes all enabled player $2$ actions into the stubborn set if it is the case.
Condition~\ref{rule:dead} ensures that at least one player $2$ action cannot be disabled solely by exploring nonstubborn actions.

\begin{exa}\label{ex:stable}
In Figure~\ref{fig:reduction-example} we see an example of a GLTS
using the previously introduced graphical notation.
Let $\goal = \{ s_8 \}$ be the set of goal states and
let $\interest{s_1}{\goal} = \{ a \}$ be a fixed set of interesting actions.
For state $s_1$ we assume $\reduction(s_1) = \{a, c\}$ as this stubborn
set satisfies the stable reduction conditions.
We satisfy~\ref{rule:game-1} since $c$ has to be fired before
we can reach the player $2$ state $s_9$.
For $s_1 \rtrans[ba] s_5$ and $s_1 \rtrans[bc] s_7$ we also have $s_1 \rtrans[ab] s_5$ and $s_1 \rtrans[cb] s_7$, so~\ref{rule:stub} is satisfied as well.
Clearly $\reduction(s_1)$ contains the interesting set
$\interest{s_1}{\goal}$ that we fixed to $\{a\}$, so~\ref{rule:reach} is satisfied.
Condition~\ref{rule:safe} is satisfied since $\reduction(s_1) \cap \en(s_1) \subseteq \safe{s_1}$.
We have that~\ref{rule:stub-init},~\ref{rule:game-2},~\ref{rule:visible}, and~\ref{rule:dead} are satisfied as well since their antecedents are not true.
Thick lines in the figure indicate transitions and states that are preserved
by a stable reduction $\reduction$, while thin lines indicates transitions and states that are removed by the same reduction.

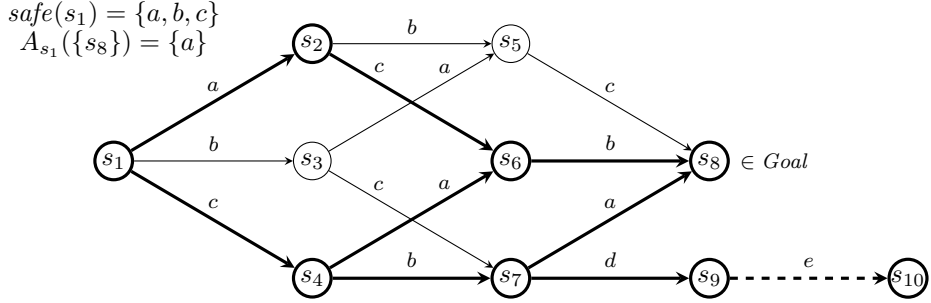
\begin{figure}[t]
    	\centering
    	\begin{tikzpicture}[font=\scriptsize,xscale=2.2,yscale=1.3]
    	    \tikzstyle{state}=[inner sep=0pt,circle,draw=black,very thick,fill=white,minimum height=5mm, minimum width=5mm,font=\small]
        	\tikzstyle{empty}=[rectangle,draw=none,font=\small]
        	\tikzstyle{reducedstate}=[inner sep=0pt,circle,draw=black,fill=white,minimum height=5mm, minimum width=5mm,font=\small]
        	\tikzstyle{every label}=[black]
        	\tikzstyle{playerArc}=[->,>=stealth,very thick]
        	\tikzstyle{reducedplayerArc}=[->,>=stealth]
        	\tikzstyle{opponentArc}=[->,>=stealth,dashed,very thick]
        	\tikzstyle{reducedopponentArc}=[->,>=stealth,dashed]
		\begin{scope}
			\node [state] at (0,0) (s1) {$s_1$};
			\node [state] at (1.2,1.2) (s2) {$s_2$};
			\node [reducedstate] at (1.2,0) (s3) {$s_3$};
			\node [state] at (1.2,-1.2) (s4) {$s_4$};
			\node [reducedstate] at (2.4,1.2) (s5) {$s_5$};
			\node [state] at (2.4,0) (s6) {$s_6$};
			\node [state] at (2.4,-1.2) (s7) {$s_7$};
			\node [state,label=right:$\in \goal$] at (3.6,0) (s8) {$s_8$};
			\node [state] at (3.6,-1.2) (s9) {$s_9$};
			\node [state] at (4.8,-1.2) (s10) {$s_{10}$};

			\node [empty] at (0,1.5) {$\safe{s_1} = \{ a, b, c \}$};
			\node [empty] at (0,1.2) {$\interest{s_1}{\{s_8\}} = \{ a \}$};

			\draw[playerArc] (s1) -- (s2) node[midway,above]{$a$} {};
			\draw[reducedplayerArc] (s1) -- (s3) node[midway,above]{$b$} {};
			\draw[playerArc] (s1) -- (s4) node[midway,above]{$c$} {};
			\draw[reducedplayerArc] (s2) -- (s5) node[midway,above]{$b$} {};
			\draw[playerArc] (s2) -- (s6) node[pos=0.3,above]{$c$} {};
			\draw[reducedplayerArc] (s3) -- (s5) node[pos=0.7,above]{$a$} {};
			\draw[reducedplayerArc] (s3) -- (s7) node[pos=0.3,above]{$c$} {};
			\draw[playerArc] (s4) -- (s6) node[pos=0.7,above]{$a$} {};
			\draw[playerArc] (s4) -- (s7) node[midway,above]{$b$} {};
			\draw[reducedplayerArc] (s5) -- (s8) node[midway,above]{$c$} {};
			\draw[playerArc] (s6) -- (s8) node[midway,above]{$b$} {};
			\draw[playerArc] (s7) -- (s8) node[midway,above]{$a$} {};
			\draw[playerArc] (s7) -- (s9) node[midway,above]{$d$} {};
			\draw[opponentArc] (s9) -- (s10) node[midway,above]{$e$} {};
		\end{scope}
    	\end{tikzpicture}
    \caption{Example of a stable reduction for a state $\mathindex{s}{1}$}%
    \label{fig:reduction-example}
\end{figure}
\end{exa}

We shall now prove the correctness of our stubborn set reduction.
We first notice the fact that if a goal state is reachable from some state,
then the state has at least one enabled action that is also in the stubborn set.

\begin{lem}\label{lemma:early-termination}
	Let $\lts = (\ltsstates, \ltslabels_1, \ltslabels_2, \ltsedges, \goal)$ be a GLTS, $\reduction$ a reduction that satisfies Conditions~\ref{rule:stub} and~\ref{rule:reach}, and $s \in \ltsstates \setminus \goal$ a state.
	If there exists $w \in \ltslabels^*$ s.t.\ $s \rtrans[w] s'$ and $s' \in \goal$ then $\reduction(s) \cap \en(s) \neq \emptyset$.
\end{lem}
\begin{proof}
	Assume that there exists $w = a_1 \cdots a_n \in \ltslabels^*$ s.t.\ $s \rtrans[w] s'$ and $s' \in \goal$.
	If $w \in \overline{\reduction(s)}^*$ then by Condition~\ref{rule:reach} we must have $s' \notin \goal$, however this contradicts our assumption.
	Therefore there must exist an action that occurs in $w$ that is in the stubborn set of $s$.
	Let $a_i \in \reduction(s)$ be the first of such an action s.t.\ for all $j$, $1 \leq j < i$, we have $a_j \notin \reduction(s)$.
	Clearly, we have $a_1 \cdots a_j \in \overline{\reduction(s)}^*$ and by Condition~\ref{rule:stub} we have $a_i \in \reduction(s) \cap \en(s)$.
\end{proof}

The correctness of stable stubborn reductions is proved by the next two
lemmas. Both lemmas are proved by induction on the depth of a winning strategy for player 1
in the game.

\begin{lem}\label{lemma1}
	Let $\lts = (\ltsstates, \ltslabels_1, \ltslabels_2, \ltsedges, \goal)$ be a GLTS and $\reduction$ a stable reduction.
If  a state $s \in \ltsstates$ is winning for player $1$ in $\lts$ then $s$ is also winning
for player $1$ in $\lts_{\reduction}$.
\end{lem}
\begin{proof}
	Assume that $s \in \ltsstates$ is a winning state for player $1$ in $\lts$.
	By definition we have that there exists a player $1$ strategy
    $\sigma$ such that for all $\pi \in \maxpaths_{\lts,\sigma}(s)$ there exists a position $i$ s.t.\ $\pi_i \in \goal$.
	By induction on $n$ we now
prove the induction hypothesis $\mathit{IH}(n)$: ``If $s$ is a winning state for player $1$ in $\lts$ with a strategy with a depth of $n$ then $s$ is a winning state for player $1$ in $\lts_{\reduction}$.''

	\emph{Base step}.
	Let $n = 0$.
	Then since $n$ is the depth at $s$ in $\lts$ we must have $s \in \goal$ and so
  $s$ is trivially a winning state for player $1$ also in $\lts_{\reduction}$.

	\emph{Induction step}.
	Let $n > 0$ and let $\sigma$ be a winning strategy with depth $n$ for $s$.
	There are three cases: (1) $\en_1(s) \neq \emptyset$ and $\en_2(s) \neq \emptyset$, (2) $\en_2(s) = \emptyset$, and (3)  $\en_1(s) = \emptyset$.
	A deadlock at $s$, i.e.\ $\en(s) = \emptyset$, cannot be the case as we otherwise have $n = 0$.

	Case (1): Let $\en_1(s) \neq \emptyset$ and $\en_2(s) \neq \emptyset$.
	We assume that
        $s$ is a winning state for player~$1$ in $\lts$ with a strategy $\sigma$ with a depth of $n$
        and we want to show that there exists a strategy $\sigma'$ s.t.\ $s$ is a winning state for player $1$ in $\lts_{\reduction}$ with $\sigma'$.
	Since $s$ is a winning state for player $1$ in $\lts$ with $\sigma$ if $s \rtrans[a] s'$ where $a \in \stratnext{\sigma}{s}$ then $s'$ is a winning state for player $1$ in $\lts$ with $m < n$ as the depth of $\sigma$ at $s'$ in $\lts$ due to property~\ref{stratB} of Lemma~\ref{lemma:depth}.
	By the induction hypothesis $s'$ is a winning state for player $1$ in $\lts_{\reduction}$ and there exists a strategy $\sigma'$ s.t.\ $\sigma'$ is a winning strategy for player $1$ at $s'$ in $\lts_{\reduction}$.
	By Condition~\ref{rule:stub-init} we know $\en_1(s) \subseteq \reduction(s)$ implying that $\sigma(s) \in \reduction(s)$.
	Player $1$ can therefore choose the same action proposed in the original game s.t.\ $\sigma'(s) = \sigma(s)$.
	From the definition of a winning strategy we have that no matter what action player $2$ chooses, the resulting state is a winning state for player $1$, and hence $s$ is a winning state for player $1$ in $\lts_{\reduction}$.

	Case (2): Let $\en_2(s) = \emptyset$.
	Assume that $s$ is a winning state for player $1$ in $\lts$ with a strategy $\sigma$ with a depth of $n$.
	We want to show that there exists a strategy $\sigma'$ s.t.\ $s$ is a winning state for player $1$ in $\lts_{\reduction}$ with $\sigma'$.
	Let $\pi \in \maxpaths_{\lts,\sigma}(s)$ be any run and $\pi_0 = s$.
	Since $s$ is a winning state for player $1$ in $\lts$ with $\sigma$ we know there exists an $m \leq n$ s.t.\ $\pi_0 \rtrans[a_1] \pi_1 \rtrans[a_2] \cdots \rtrans[a_m] \pi_m$ and $\pi_m \in \goal$.
	Let $w = a_1 \cdots a_m$.
	We start by showing that there exists $i$, $1 \leq i \leq m$, such that $a_i \in \reduction(s)$.
	Assume that $w \in \overline{\reduction(s)}^*$ is true.
	Then we have $\pi_m \notin \goal$ due to Condition~\ref{rule:reach}, a contradiction.
	Therefore there must exist $i$, $1 \leq i \leq m$,  s.t.\ $a_i \in \reduction(s)$.
	Let $i$ be minimal in the sense that for all $j$, $1 \leq j < i$, we have $a_j \notin \reduction(s)$.
	We can then divide $w$ s.t.\ $w = va_{i}u$, $v \in \overline{\reduction(s)}^*$ and we have $s \rtrans[a_i] s'_0 \rtrans[v] \pi_i \rtrans[u] \pi_m$ due to Condition~\ref{rule:stub} as well as $s \rtrans[a_i][\reduction] s'_0$.
	There are two subcases: (2.1) $a_i \in \safe{s}$ or (2.2) $a_i \notin \safe{s}$.
	\begin{itemize}
		\item Case (2.1): Let $a_i \in \safe{s}$.
		For all $1 \leq j < i$ we have $\en_2(\pi_j) = \emptyset$ due to $i$ being minimal and Condition~\ref{rule:game-1}.
		From that, if $a_i \in \safe{s}$ then for all intermediate states in $s \rtrans[a_{i}v] \pi_i$ we only have player $1$ states otherwise $a_i$ is not a safe action due to the definition of safe actions.
		We have that $s'_0$ is a player $1$ state and let $v = a_1 a_2 \cdots a_{i-1}$ s.t.\ $s'_0 \rtrans[a_1] s'_1 \rtrans[a_2] \cdots \rtrans[a_{i-1}] \pi_i$ and for all $k$, $1 \leq k < i-1$, we have $\en_2(s'_k) = \emptyset$.
		Let $\sigma''$ be defined such that for all $j$, $0 < j < i-1$, we have $\sigma''(s'_{j-1}) = a_j$, and let $\sigma''$ from $\pi_i$ be defined as $\sigma$.
		Clearly, $\sigma''$ is a winning strategy for player $1$ at $s'_0$ in $\lts$.
		Due to property~\ref{stratB} of Lemma~\ref{lemma:depth} the depth of $\sigma''$ at $\pi_i$ in $\lts$ is at most $k \leq n-i$.
		Since $\lts$ is deterministic by following the strategy $\sigma''$ from $s'_0$ we always reach $\pi_i$ in $i-1$ actions.
		From this we can infer that the depth of $\sigma''$ at $s'_0$ in $\lts$ is at most $k+i-1$ which is clearly smaller than $n$.
		Therefore $s'_0$ is a winning state for player $1$ in $\lts$ with at most $k+i-1 < n$ as the depth of $\sigma''$ at $s'_0$ in $\lts$.
		By the induction hypothesis $s'_0$ is a winning state for player $1$ in $\lts_{\reduction}$ and there exists a strategy $\sigma'$ s.t.\ $\sigma'$ is a winning strategy for player $1$ at $s'_0$ in $\lts_{\reduction}$.
		Player $1$ can then choose $a_i$ in the reduced game such that $\sigma'(s) = a_i$ and $s$ is a winning state for player $1$ in $\lts_{\reduction}$.
		\item Case (2.2): Let $a_i \notin \safe{s}$.
		Since $a_i \notin \safe{s}$ we have $\reduction(s) \cap \en_1(s) \nsubseteq \safe{s}$ and $\en_1(s) \subseteq \reduction(s)$ by Condition~\ref{rule:safe}.
		If $s \rtrans[\sigma(s)] s'$ then $s'$ is a winning state for player $1$ in $\lts$ with $m < n$ as the depth of $\sigma$ at $s'$ in $\lts$, following property~\ref{stratB} of Lemma~\ref{lemma:depth}.
		By the induction hypothesis $s'$ is a winning state for player $1$ in $\lts_{\reduction}$ and there exists a strategy $\sigma'$ s.t.\ $\sigma'$ is a winning strategy for player $1$ at $s'$ in $\lts_{\reduction}$.
		Player $1$ can choose the same action proposed in the original game such that $\sigma'(s) = \sigma(s)$ and $s$ is a winning state for player $1$ in $\lts_{\reduction}$.
	\end{itemize}

	Case (3): Let $\en_1(s) = \emptyset$.
	Assume that $s$ is a winning state for player $1$ in $\lts$ with $\sigma$ as the winning strategy.
	We want to show that there exists a strategy $\sigma'$ s.t.\ $s$ is a winning state for player $1$ in $\lts_{\reduction}$ with $\sigma'$.
	Since $\en_1(s) = \emptyset$ we have $\sigma(s) = \sigma'(s) = \bot$ by the definition of strategies.
	We have from the definition of a winning strategy that no matter what action player~$2$ chooses,
         the resulting state is a winning state for player $1$.
	What remains to be shown is that at least one enabled player $2$ action is included in $\reduction(s)$.
	As $\en_2(s) \neq \emptyset$, due to Condition~\ref{rule:dead} we get that there exists $a \in \en_2(s) \cap \reduction(s)$, and this last case is also established.
\end{proof}

\begin{lem}\label{lemma2}
	Let $\lts = (\ltsstates, \ltslabels_1, \ltslabels_2, \ltsedges, \goal)$ be a GLTS and $\reduction$ a stable reduction.
If a state $s \in \ltsstates$ is winning for player $1$ in $\lts_{\reduction}$ then $s$ is
also  winning for player $1$ in $\lts$.
\end{lem}
\begin{proof}
	Assume that $s \in \ltsstates$ is a winning state for player $1$ in $\lts_{\reduction}$.
	By definition we have that there exists a strategy
 $\sigma$ s.t.\ for all $\pi \in \maxpaths_{\lts_{\reduction},\sigma}(s)$ there exists a position $i$ s.t.\ $\pi_i \in \goal$.
	Let $\sigma$ be fixed for the remainder of the proof.
	Let $n$ be the depth of $\sigma$ at $s$ in $\lts_{\reduction}$.
	By induction on $n$ we prove the induction hypothesis $\mathit{IH}(n)$:
``If $s$ is a winning state for player $1$ in $\lts_{\reduction}$ with a strategy with a depth of $n$ then $s$ is a winning state for player $1$ in $\lts$.''

	\emph{Base step}.
	If $n = 0$ then since $n$ is the depth at $s$ in $\lts_{\reduction}$ we must
  have $s \in \goal$, implying that $s$ is a winning state for player $1$ also in $\lts$.

	\emph{Induction step}.
	Let $n > 0$ and let $s$ be a winning state for player $1$ in
        $\lts_{\reduction}$ with a strategy with a depth of $n$.
	There are three cases: (1) $\en_1(s) \cap \reduction(s) \neq \emptyset$ and $\en_2(s) \cap \reduction(s) \neq \emptyset$, (2) $\en_2(s) \cap \reduction(s) = \emptyset$, and (3)  $\en_1(s) \cap \reduction(s) = \emptyset$.
	A deadlock at $s$ in $\lts_{\reduction}$ such that $\en(s) \cap \reduction(s) = \emptyset$
        is not possible as otherwise we have the case where $n = 0$.

	Case (1): Let $\en_1(s) \cap \reduction(s) \neq \emptyset$ and $\en_2(s) \cap \reduction(s) \neq \emptyset$.
	We assume that $s$ is a winning state for player $1$ in $\lts_{\reduction}$ with a strategy $\sigma$ with a depth of $n$.
	We want to show that there exists a strategy $\sigma'$ s.t.\ $s$ is a winning state for player $1$ in $\lts$ with $\sigma'$.
	Since $s$ is a winning state for player $1$ in $\lts_{\reduction}$ with $\sigma$,
whenever $s \rtrans[\sigma(s)][\reduction] s'$ or $s \rtrans[a][\reduction] s'$ where $a \in \en_2(s) \cap \reduction(s)$ then $s'$ is a winning state for player $1$ in $\lts_{\reduction}$ with $m < n$ as the depth of $\sigma$ at $s'$ in $\lts_{\reduction}$, following property~\ref{stratB} of Lemma~\ref{lemma:depth}.
	By the induction hypothesis $s'$ is a winning state for player $1$ in $\lts$ and there exists a strategy $\sigma'$ s.t.\ $\sigma'$ is a winning strategy for player $1$ at $s'$ in $\lts$.
	Since $s$ is a mixed state in $\lts_{\reduction}$ then $s$ must also be a mixed state in $\lts$ due to $\rtrans[][\reduction] \subseteq \ltsedges$. This implies that $s \rtrans[\sigma(s)] s'$.
	Therefore player $1$ can choose the same action proposed in the reduced game such that $\sigma'(s) = \sigma(s)$.
	Furthermore we have $\en_2(s) \cap \reduction(s) = \en_2(s)$ from Condition~\ref{rule:stub-init}.
	From this we can conclude that $s$ is a winning state for player $1$ in $\lts$ with strategy $\sigma'$.

	Case (2): Let $\en_2(s) \cap \reduction(s) = \emptyset$.
	Assume that $s$ is a winning state for player $1$ in $\lts_{\reduction}$ with a strategy $\sigma$ with a depth of $n$.
	We want to show that there exists a strategy $\sigma'$ s.t.\ $s$ is a winning state for player $1$ in $\lts$ with $\sigma'$.
	Since $s$ is a winning state for player $1$ in $\lts_{\reduction}$ with $\sigma$ we have $s \rtrans[\sigma(s)][\reduction] s'$ and $s'$ is a winning state for player $1$ in $\lts_{\reduction}$ with $m < n$ as the depth of $\sigma$ at $s'$ in $\lts_{\reduction}$, following property~\ref{stratB} of Lemma~\ref{lemma:depth}.
	By the induction hypothesis $s'$ is a winning state for player $1$ in $\lts$ and there exists a strategy $\sigma'$ s.t.\ $\sigma'$ is a winning strategy for player $1$ at $s'$ in $\lts$.
	Trivially we have that $s \rtrans[\sigma(s)] s'$ since we have $\rtrans[][\reduction] \subseteq \ltsedges$.
	Therefore player $1$ can choose the same action proposed in the reduced game $\sigma'(s) = \sigma(s)$.
	Next we show by contradiction that $s$ is a player~$1$ state also in $\lts$.
	Assume $\en_2(s) \neq \emptyset$, i.e.\ that $s$ is a mixed state in $\lts$.
	From this we can infer by Condition~\ref{rule:stub-init} that $\en(s) \subseteq \reduction(s)$ and $\en_2(s) \cap \reduction(s) \neq \emptyset$, which is a contradiction.
	Therefore we have $\en_2(s) = \emptyset$, i.e.\ $s$ is a player $1$ state also in $\lts$, and $s$ is a winning state for player $1$ in $\lts$ with strategy $\sigma'$.

	Case (3): Let $\en_1(s) \cap \reduction(s) = \emptyset$.
	Assume that $s$ is a winning state for player $1$ in $\lts_{\reduction}$ with a strategy $\sigma$ with a depth of $n$.
	We want to show that there exists a strategy $\sigma'$ s.t.\ $s$ is a winning state for player $1$ in $\lts$ with $\sigma'$.
	Since $\en_1(s) \cap \reduction(s) = \emptyset$ then we have $\sigma(s) = \bot$.
	Furthermore, we have $\en_1(s) = \emptyset$ since otherwise with Condition~\ref{rule:stub-init} we will be able to infer that $\en_1(s) \cap \reduction(s) \neq \emptyset$, which is a contradiction.
	We define $\sigma'=\sigma$.

What remains to be shown is that $s$ is a winning state for player $1$ in $\lts$.
For the sake of contradiction assume that this is not the case, i.e.
that there exists $\pi \in \maxpaths_{\lts,\sigma'}(s)$  such that
\[s=\pi_0 \rtrans[a_1] \pi_1 \rtrans[a_2] \pi_2 \rtrans[a_3] \cdots\]
and $\pi_i \notin \goal$ for all positions $i$.
We shal first argue that $a_1 \notin \reduction(s)$.
If this is not the case, then $\pi_0 \rtrans[a_1][\reduction] \pi_1$ also in the reduced game $\lts_{\reduction}$.
Due to our assumption that $s = \pi_0$ is a winning state for player $1$ in $\lts_{\reduction}$ and $a_1 \in \ltslabels_2$, we know that also $\pi_1$ is a winning state for player $1$ in $\lts_{\reduction}$ with $m < n$ as the depth of $\sigma$ at $\pi_1$ in $\lts_{\reduction}$, following property~\ref{stratB} of Lemma~\ref{lemma:depth}.
By the induction hypothesis $\pi_1$ is a winning state for player $1$ in $\lts$,
which contradicts the existence of the maximal path $\pi$ with no goal states.

Let us so assume that $a_1 \notin \reduction(s)$.
Let $j > 1$ be the smallest index such that $a_j \in \reduction(s)$
and $a_1a_2\cdots a_{j-1}a_j \in \ltslabels_2^*$. Such index must
exist because of the following case analysis.
\begin{itemize}
\item Either the sequence $a_1a_2\cdots$ contains an action that belongs to
$\ltslabels_1$ (we note that because of our assumption $a_1 \notin \ltslabels_1$).
Due to Condition~\ref{rule:game-2} there must exist
an action $a_j$ that is stubborn in $s$ and let $j > 1$ be the smallest index
such that $a_j \in \reduction(s)$. As $a_j$ is the first action that is stubborn,
we get that $a_1a_2\cdots a_{j-1} a_j \in \ltslabels_2^*$ as otherwise
the existence of $i \leq j$ where $a_i \in \ltslabels_1$ contradicts
the minimality of $j$ due to Condition~\ref{rule:game-2}.
\item
Otherwise the sequence $a_1a_2\cdots$ consists solely of actions from $\ltslabels_2$.
If the sequence contains a stubborn action then we are done and similarly, if the sequence
is finite and ends in a deadlock, we get by Condition~\ref{rule:dead} that there must be
an $j > 1$ where  $a_j \in \reduction(s)$ as required.
The last option is that the sequence $a_1a_2\cdots$ is infinite and does not contain any stubborn action.
By Condition~\ref{rule:dead} there exists $a \in \en_2(s) \cap \reduction(s)$ such that
for all $i > 0$ we have $s \rtrans[a_1 \cdots a_i] \pi_i \rtrans[a] \pi_i'$ and then
by Condition~\ref{rule:stub} we get $s \rtrans[a] \pi_0' \rtrans[a_1 \cdots a_i] \pi_i'$.
This implies that 
from $\pi_0'$ we can also execute the
infinite sequence of actions $a_1a_2\cdots$ while Condition~\ref{rule:visible} guarantees
that none of the states visited during this execution is a goal state.
Hence the state $\pi_0'$ must be losing for player 1 in $\lts$, which however
contradicts that by induction hypothesis $\pi_0'$ is winning for player 1 in $\lts$
as $s \rtrans[a] \pi_0'$ with $a \in \reduction(s)$ and the depth of
player 1 winning strategy
at $\pi_0'$ in $\lts_{\reduction}$ is smaller than the depth at
$s$ in $\lts_{\reduction}$.
Hence there cannot be any infinite sequence of nonstubborn actions starting from $s$.
\end{itemize}
As we have now established that there is the smallest index $j > 1$ such that
$a_j \in \reduction(s)$ and $a_1a_2\cdots a_{j-1}a_j \in \ltslabels_2^*$,
the minimality of $j$ implies that
$a_1a_2\cdots a_{j-1} \in \overline{\reduction(s)}^*$.
This means that we can apply Condition~\ref{rule:stub} and conclude
that there exists a maximal run $\pi'$ given by
\[s \rtrans[a_j] s' \rtrans[a_1 a_2 \cdots a_{j-1}] \pi_j \rtrans[a_{j+1}] \pi_{j+1}
\rtrans[a_{j+2}]  \cdots\]
that is from $\pi_j$ identical to the run of $\pi$. Hence
$\pi_i \notin \goal$ for all $i \geq j$. We notice that
also the intermediate states in the prefix of the run $\pi'$ may not be goal states,
which is implied by Condition~\ref{rule:visible} and the fact that
$a_1 \in \en_2(s)$, $a_1a_2\cdots a_{j-1}a_j \in \ltslabels_2^*$, and $a_1 \notin \reduction(s)$.
However, as $a_j \in \reduction(s)$ we get
$s \rtrans[a_j][\reduction] s'$ and because $a_j \in \ltslabels_2$ we know
that $s'$ is a winning state for player $1$ in $\lts_{\reduction}$ with $m < n$ as the depth of $\sigma$ at $s'$ in $\lts_{\reduction}$, following property~\ref{stratB} of Lemma~\ref{lemma:depth}.
By the induction hypothesis $s'$ is a winning state for player $1$ in $\lts$,
which contradicts the existence of a maximal run from $s'$ that contains no goal states.
Hence the proof of Case (3) is finished.
\end{proof}

We can now present the main theorem showing that stable reductions preserve
the winning strategies of both players in the game.

\begin{thm}[Strategy Preservation for GLTS]\label{theorem:preservation-1}
	Let $\lts = (\ltsstates, \ltslabels_1, \ltslabels_2, \ltsedges, \goal)$ be a GLTS and $\reduction$ a stable reduction.
A state $s \in \ltsstates$ is winning for player $1$ in $\lts$ iff
$s$ is winning for player $1$ in $\lts_{\reduction}$.
\end{thm}
\begin{proof}
	Follows from Lemma~\ref{lemma1} and~\ref{lemma2}.
\end{proof}

\begin{rem}\label{rem:error}
In~\cite{boenneland2019partial} we omitted Condition~\ref{rule:visible}
from the definition of stable reduction and this implied that Lemma~\ref{lemma2}
(as it was stated in~\cite{boenneland2019partial}) did not hold.
We illustrate this in Figure~\ref{fig:counter} where all actions are player $2$ actions
and the goal state is $s_2$. Clearly, player $1$ does not have a winning strategy
as player $2$ can play the action $b$ followed by $a$ and reach the deadlock state $s_4$
without visiting the goal state.
The stubborn set $\reduction(s_1) = \{ a \}$ on the other hand satisfies all
conditions of the stable reduction, except for~\ref{rule:visible}, however,
it breaks Lemma~\ref{lemma2} because in the reduced system the action $b$
in $s_1$ is now exluded and the (only) stubborn action $a$ for the environment
brings us to a goal state.
It is therefore the case that in the original game $s_1$ is not a winning state for player $1$
but in the reduced game it is.
The extra Condition~\ref{rule:visible} introduced in this article forces us to include
all enabled actions in $s_1$ into the stubborn set, and hence the validity of
Lemma~\ref{lemma2} is recovered.

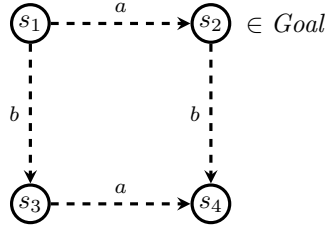
\begin{figure}[t]
    	\centering
    	\begin{tikzpicture}[font=\scriptsize,xscale=2,yscale=2]
    	    \tikzstyle{state}=[inner sep=0pt,circle,draw=black,very thick,fill=white,minimum height=5mm, minimum width=5mm,font=\small]
        	\tikzstyle{empty}=[rectangle,draw=none,font=\small]
        	\tikzstyle{reducedstate}=[inner sep=0pt,circle,draw=black,fill=white,minimum height=5mm, minimum width=5mm,font=\small]
        	\tikzstyle{every label}=[black]
        	\tikzstyle{playerArc}=[->,>=stealth,very thick]
        	\tikzstyle{reducedplayerArc}=[->,>=stealth]
        	\tikzstyle{opponentArc}=[->,>=stealth,dashed,very thick]
        	\tikzstyle{reducedopponentArc}=[->,>=stealth,dashed]

		\begin{scope}
			\node [state] at (0,0) (s1) {$s_1$};
			\node [state] at (1.2,0) (s2) {$s_2$};
			\node [state] at (0,-1.2) (s3) {$s_3$};
			\node [state] at (1.2,-1.2) (s4) {$s_4$};
			\node [empty] at (1.7,0) {$\in \goal$};

			\draw[opponentArc] (s1) -- (s2) node[midway,above]{$a$} {};
			\draw[opponentArc] (s1) -- (s3) node[midway,left]{$b$} {};
			\draw[opponentArc] (s2) -- (s4) node[midway,left]{$b$} {};
			\draw[opponentArc] (s3) -- (s4) node[midway,above]{$a$} {};
		\end{scope}
    	\end{tikzpicture}
    \caption{Example showing the importance of Condition~\ref{rule:visible}}%
    \label{fig:counter}
\end{figure}
\end{rem}

Finally, we notice that for non-mixed games we can simplify the conditions of stable
reductions by removing the requirement on safe actions.

\begin{thm}[Strategy Preservation for Non-Mixed GLTS]\label{theorem:preservation-2}
	Let $\lts = (\ltsstates, \ltslabels_1, \ltslabels_2, \ltsedges, \goal)$ be a non-mixed GLTS and $\reduction$ a stable reduction with Condition~\ref{rule:safe} excluded.
A state $s \in \ltsstates$ is winning for player $1$ in $\lts$ iff
$s$ is winning for player $1$ in $\lts_{\reduction}$.
\end{thm}
\begin{proof}
In Lemma~\ref{lemma2} the condition~\ref{rule:safe} is not used at all.
In Lemma~\ref{lemma1} the subcase (2.2) is the only one that relies on~\ref{rule:safe}.
Because there are no mixed states, the arguments in subcase (2.1) are valid
irrelevant of whether $a_i$ is safe or not.
\end{proof}

%% file: sections/petrigames.tex
\section{Stable Reductions on Petri Net Games}
We now introduce the formalism of Petri net games and show
how to algorithmically construct stable reductions in a syntax-driven
manner.

\begin{defi}[Petri Net Game]
    A \PNG is a tuple $N = \petrituple$ where
\begin{itemize}
\item $\places$ and
$\transitions = \transitions_1 \uplus \transitions_2$
are finite sets of places and transitions, respectively, such that
$\places \cap \transitions = \emptyset$ and where transitions
are partitioned into player $1$ and player $2$ transitions,
\item
$\weights: (\places \times \transitions) \cup (\transitions \times \places) \rightarrow \mathbb{N}^0$ is a weight function for regular arcs, and
\item $\inhib: (\places \times \transitions) \rightarrow \mathbb{N}^{\infty}$ is a weight function for inhibitor arcs.
\end{itemize}
    A \emph{marking} $M$ is a function $M: P \to \mathbb{N}^0$ and
    $\markings$ denotes the set of all markings for $N$.
\end{defi}
For the rest of this section,
let $N = \petrituple$ be a fixed \PNG such that
$\transitions = \transitions_1 \uplus \transitions_2$.
Let us first fix some useful notation.
For a place or transition $x$, we denote the \emph{preset} of $x$ as $\preset{x} = \{ y \in \places \cup \transitions \mid \weights(y,x) > 0 \}$, and the \emph{postset} of $x$ as $\postset{x} = \{ y \in \places \cup \transitions \mid \weights(x,y) > 0 \}$.
For a transition $t$, we denote the \emph{inhibitor preset} of $t$ as $\inhibpreset{t} = \{ p \in \places \mid \inhib(p,t) \neq \infty \}$, and the \emph{inhibitor postset} of a place $p$ as $\inhibpostset{p} = \{ t \in \transitions \mid \inhib(p,t) \neq \infty \}$.
For a place $p$ we define the \emph{increasing preset} of $p$, containing all transitions that increase the number of tokens in $p$, as $\preincr{p} = \{ t \in \preset{p} \mid \weights(t,p) > \weights(p,t) \}$, and similarly the \emph{decreasing postset} of $p$ as $\postdecr{p} = \{ t \in \postset{p} \mid \weights(t,p) < \weights(p,t) \}$.
For a transition $t$ we define the \emph{decreasing preset} of $t$, containing all places that have their number of tokens decreased by $t$, as $\predecr{t} = \{ p \in \preset{t} \mid \weights(p,t) > \weights(t,p) \}$, and similarly the \emph{increasing postset} of $t$ as $\postincr{t} = \{ p \in \postset{t} \mid \weights(p,t) < \weights(t,p) \}$.
For a set $X$ of either places or transitions, we extend the notation as $\preset{X} = \bigcup_{x \in X} \preset{x}$ and $\postset{X} = \bigcup_{x \in X} \postset{x}$, and similarly for the other operators.

A Petri net $N = \petrituple$ defines a GLTS $\lts(N) = (\ltsstates, \ltslabels_1, \ltslabels_2, \ltsedges, \goal)$ where
\begin{itemize}
\item $\ltsstates = \markings$ is the set of all markings,
\item
$\ltslabels_1 = \transitions_1$ is the set of player 1 actions,
\item
$\ltslabels_2 = \transitions_2$ is the set of player 2 actions,
\item $M \rtrans[t] M'$ whenever for all $p \in P$ we have $M(p) \geq \weights(p,t)$, $M(p) < \inhib(p,t)$ and $M'(p) = M(p) - \weights(p,t) + \weights(t,p)$, and
\item $\goal \in \markings$ is the set of goal markings, described by a simple reachability
logic formula defined below.
\end{itemize}

\noindent
Let $E_N$ be the set of marking expressions in $N$
given by the abstract syntax (here $e$ ranges over $E_N$):
    \[e ::= c \mid p \mid e_1 \oplus e_2\]
where $c \in \mathbb{N}^0$, $p \in \places$, and $\oplus \in \{ +, -, * \}$.
An expression $e \in E_N$ is evaluated relative to a marking $M \in \markings$ by the function $\eval_M : E_N \to \mathbb{Z}$
where $\eval[M][c] = c$, $\eval[M][p] = M(p)$ and
$\eval[M][e_{1} \oplus e_{2}] = \eval[M][e_{1}] \oplus \eval[M][e_{2}]$.

In Table~\ref{tab:incr-and-decr} we define the functions $\mathit{incr_M}: E_N \to 2^T$ and $\mathit{decr_M}: E_N \to 2^T$ that,
given an expression $e \in E_N$, return the set of transitions
that can (when fired) increase resp.\ decrease the evaluation of $e$.
We note that transitions in $\mathit{incr_M}(e)$ and $\mathit{decr_M}(e)$
are not necessarily enabled in $M$, however, due to Lemma~\ref{lemma:incr},
if a transition firing increases the evaluation of $e$ then the transition
must be in $\mathit{incr_M}(e)$,
and similarly for $\mathit{decr_M}(e)$.

\begin{lemC}[\cite{boenneland2018start}]\label{lemma:incr}
	Let $N = \petrituple$ be a Petri net and $M \in \markings$ a marking.
	Let $e \in E_N$ and let $M \rtrans[w] M'$ where $w=t_1t_2\dots t_n \in T^*$.
	\begin{itemize}
		\item
		If $\eval[M][e] < \eval[M'][e]$ then there is $i$, $1 \leq i \leq n$, such that $t_i \in \incr[M][e]$.
		\item
		If $\eval[M][e] > \eval[M'][e]$ then there is $i$, $1 \leq i \leq n$, such that $t_i \in \decr[M][e]$.
	\end{itemize}
\end{lemC}

\begin{table}[t]
\centering
\scalebox{0.94}{%
    \def\arraystretch{1.2}
    \begin{tabular}{lll}
    Expression $e$ & $\incr[M][e]$  & $\decr[M][e]$ \\ \toprule
    $c$ & $\emptyset$ & $\emptyset$ \\
    $p$ & $\preincr{p}$ & $\postdecr{p}$ \\
    $e_1 + e_2$ & $\incr[M][e_1] \cup \incr[M][e_2]$ & $\decr[M][e_1] \cup \decr[M][e_2]$ \\
    $e_1 - e_2$ & $\incr[M][e_1] \cup \decr[M][e_2]$ & $\decr[M][e_1] \cup \incr[M][e_2]$ \\
    $e_1 \cdot e_2$ & $\incr[M][e_1] \cup \decr[M][e_1] \cup {}$
                & $\incr[M][e_1] \cup \decr[M][e_1] \cup {}$ \\
    & \;$ \incr[M][e_2] \cup \decr[M][e_2]$ & \;$ \incr[M][e_2] \cup \decr[M][e_2]$
    \end{tabular}
}
    \caption{Increasing and decreasing transitions for expression $e \in \mathindex{E}{N}$}%
    \label{tab:incr-and-decr}
\end{table}

We can now define the set of reachability formulae $\Phi_N$
that evaluate over the markings in $N$ as follows:
\[\varphi ::= \true \mid \false \mid t \mid e_1 \bowtie e_2  \mid
 \deadlock \mid \varphi_{1} \land \varphi_{2} \mid \varphi_{1} \lor \varphi_{2}
  \mid \neg\varphi\]
where $e_1,e_2 \in E_N$, $t \in \transitions$ and $\bowtie$ $\in$ $\{<, \leq, =, \neq, >, \geq\}$.

The satisfaction relation for a formula $\varphi \in \Phi_N$
in a marking $M$ is defined as expected:
\begin{align*}
	&M \models \true && \\
    &M \models t && \mbox{iff } t \in \en(M) \\
    &M \models e_1 \bowtie e_2 && \mbox{iff } \eval[M][e_1] \bowtie \eval[M][e_2] \\
	&M \models \deadlock &&\mbox{iff } \en(M) = \emptyset \\
	&M \models \varphi_1 \land \varphi_2 &&\mbox{iff } M \models \varphi_1 \mbox{ and } M \models \varphi_2 \\
	&M \models \varphi_1 \lor \varphi_2 &&\mbox{iff } M \models \varphi_1 \mbox{ or } M \models \varphi_2 \\
	&M \models \neg \varphi&&\mbox{iff } M \not\models \varphi
\end{align*}

We want to be able to preserve at least one execution to the set
$\goal=\{ M \in \markings \mid M \models \varphi \}$ for
a given formula $\varphi$ describing the set of goal markings.
In order to achieve this, we define the set of interesting transitions
$\interesting[M][{\varphi}]$
for a formula $\varphi$ so that any firing sequence of transitions
from a marking that does not satisfy $\varphi$ leading to a marking
that satisfies $\varphi$ must contain at least one interesting transition.
Table~\ref{tab:bool-formula} provides the definition of
$\interesting[M][{\varphi}]$ that is
similar to the one presented in~\cite{boenneland2018start} for the
non-game setting, except for the conjunction where we in our setting
use Equation~(\ref{eq:interest-and}) that provides an
 optimisation for Condition~\ref{rule:safe} and possibly ends with a
smaller set of interesting transitions.

\begin{table}[t]
\centering
\scalebox{0.94}{
\def\arraystretch{1.2}
\begin{tabular}{lll}
    $\varphi$ & $\interest{M}{{\varphi}}$ & $\interesting[M][{\neg\varphi}]$ \\ \toprule

    $\mathit{deadlock}$ & $t \cup \postdecr{(\preset{t})} \cup \preincr{(\inhibpreset{t})}$ for some selected $t \in \en(M)$ & $\emptyset$  \\[2mm]

    $t$ & \pbox{20cm}{$\preincr{p}$ for some selected $p \in \preset{t}$ where $M(p) < \weights(p,t)$, or \\ $\postdecr{p}$ for some selected $p \in\inhibpreset{t}$ where $M(p) \geq \inhib(p,t)$} & $\postdecr{(\preset{t})}\cup \preincr{(\inhibpreset{t})}$ \\[3mm]

    $e_ 1 < e_ 2$ & $\decr[M][e_1] \cup \incr[M][e_2]$ & $\interesting[M][{e_1 \geq e_2}]$ \\

    $e_ 1 \leq e_ 2$ & $\decr[M][e_1] \cup \incr[M][e_2]$ & $\interesting[M][{e_1 > e_2}]$ \\

    $e_ 1 > e_ 2$ & $\incr[M][e_1] \cup \decr[M][e_2]$ & $\interesting[M][{e_1 \leq e_2}]$ \\

    $e_ 1 \geq e_ 2$ & $\incr[M][e_1] \cup \decr[M][e_2]$ & $\interesting[M][{e_1 < e_2}]$ \\[2mm]

    $e_ 1 = e_ 2$ & \pbox{10cm}{
$\decr[M][e_1] \cup \incr[M][e_2]$ if $\eval[M][e_1] > \eval[M][e_2]$
\\
$\incr[M][e_1] \cup \decr[M][e_2]$ if $\eval[M][e_1] < \eval[M][e_2]$}
& $\interesting[M][{e_1 \neq e_2}]$ \\[3mm]

    $e_ 1 \neq e_ 2$ & $\incr[M][e_1] \cup \decr[M][e_1] \cup\incr[M][e_2] \cup \decr[M][e_2]$ & $\interesting[M][{e_1 = e_2}]$ \\

    $\varphi_1 \land \varphi_2$ & Defined in Equation~(\ref{eq:interest-and}) & $\interesting[M][{\neg\varphi_1\lor\neg\varphi_2}]$ \\

    $\varphi_1 \lor \varphi_2$ & $\interesting[M][{\varphi_1}] \cup \interesting[M][{\varphi_2}]$ & $\interesting[M][{\neg\varphi_1\land\neg\varphi_2}]$
    \end{tabular}
}
    \caption{Interesting transitions of $\varphi$ (assuming $M \not\models \varphi$, otherwise
$\interesting[M][{\varphi}] = \emptyset$)}%
    \label{tab:bool-formula}
\end{table}

\begin{equation}\label{eq:interest-and}
\interest{M}{{\varphi_1 \land \varphi_2}} =
\begin{cases}
\interest{M}{{\varphi_1}} & \text{if } M\models\varphi_2\\
\interest{M}{{\varphi_2}} & \text{if } M\models\varphi_1\\
\interest{M}{{\varphi_1}} & \text{if } M\not\models\varphi_1 \text{ and } \interest{M}{{\varphi_1}} \subseteq \safe{M}\\
\interest{M}{{\varphi_2}} & \text{if } M\not\models\varphi_2 \text{ and } \interest{M}{{\varphi_2}} \subseteq \safe{M}\\
\interest{M}{{\varphi_i}} & \text{otherwise where } i \in \{1,2\}
\end{cases}
\end{equation}

The desired property of the set of interesting transitions is formulated
below.
\begin{lem}\label{lemma:reach}
	Let $N = \petrituple$ be a Petri net, $M \in \markings$ a marking, and $\varphi \in \Phi_N$ a formula.
	If $M \not\models \varphi$ and $M \rtrans[w] M'$ where $w \in \overline{\interest{M}{{\varphi}}}^*$
 then $M' \not\models \varphi$.
\end{lem}
\begin{proof}
	Assume that $M \not\models \varphi$.
	The proof proceeds by structural induction on $\varphi$.
	All cases, with the exception of $\varphi_1 \land \varphi_2$, are proved in Lemma 2 presented in~\cite{boenneland2018start}.
Let $\varphi = \varphi_1 \land \varphi_2$.
		There are five subcases defined by Equation~\ref{eq:interest-and}:
		(1) $M\models\varphi_2$, (2) $M\models\varphi_1$, (3) $M\not\models\varphi_1$ and $\interest{M}{\varphi_1} \subseteq \safe{M}$, (4) $M\not\models\varphi_2$ and $\interest{M}{\varphi_2} \subseteq \safe{M}$, and (5) the default case.
		\begin{itemize}
			\item Case (1): Let $M\models\varphi_2$.
			Since we have $M \not\models \varphi$ and $M\models\varphi_2$ we must therefore have that $M \not\models \varphi_1$ by the semantics of $\varphi$.
			By Equation~\ref{eq:interest-and}, since $M\models\varphi_2$, we have $\interest{M}{\varphi_1 \land \varphi_2} = \interest{M}{\varphi_1}$.
			By the induction hypothesis this implies $M' \not\models \varphi_1$, and from this and the semantics of $\varphi$ we have $M' \not\models\varphi$.
			\item Case (2): Let $M\models\varphi_1$.
			This case is symmetric to Case (1) and follows the same approach.
			\item Case (3): Let $M\not\models\varphi_1$ and $\interest{M}{\varphi_1} \subseteq \safe{M}$.
			By Equation~\ref{eq:interest-and} we have $\interest{M}{\varphi_1 \land \varphi_2} = \interest{M}{\varphi_1}$.
			By the induction hypothesis this implies $M' \not\models \varphi_1$, and from this and the semantics of $\varphi$ we have $M' \not\models\varphi$.
			\item Case (4): Let $M\not\models\varphi_2$ and $\interest{M}{\varphi_2} \subseteq \safe{M}$.
			This case is symmetric to Case (3) and follows the same approach.
			\item Case (5): Default case.
			We have $M\not\models\varphi_1$ and $M\not\models\varphi_2$ due to Equation~\ref{eq:interest-and} and $\interest{M}{\varphi_1 \land \varphi_2} = \interest{M}{\varphi_i}$ for some $i \in \{1,2\}$.
			By the induction hypothesis this implies $M' \not\models \varphi_i$, and from this and the semantics of $\varphi$ we have $M' \not\models\varphi$.
            \qedhere
		\end{itemize}
\end{proof}

\begin{algorithm}[t]
        \SetKwInOut{Input}{input}
        \SetKwInOut{Output}{output}
        \Input{$N=\petrituple$ with $M \in \markings$ and a formula
$\varphi \in \Phi_N$}
        \Output{If there is $w \in \ltslabels_2^*$ s.t. $M \rtrans[w] M'$ and $M' \models \varphi$
then the algorithm returns \emph{true}.}

        We assume that all negations in $\varphi$ are only in front of atomic
        propositions (if not, we can use De Morgan's laws in order to guarantee this).

        $\ubound(x) := \infty$ for all $x \in \places \cup \transitions_2$;\label{line:init-inf}

        $\ubound(p) := M(p)$ for all $p \in P$ such that
               $\weights(p,t) \geq \weights(t,p)$ for every
               $t \in \preset{p} \cap \transitions_2$;\label{line:cond1}

        \Repeat{$\ubound(x)$ stabilises for all $x \in \places \cup \transitions_2$}{
          \ForEach{$t \in \transitions_2$}{
                $\displaystyle \ubound(t) := \min_{ p \in  \predecr{t}}
                  \lfloor \frac{\ubound(p)}{\weights(p,t)-\weights(t,p)}\rfloor$\label{line:tup}
          }
          \ForEach{$p \in \places$}{
                $\displaystyle \ubound(p) := M(p) + \sum_{\substack{t \in\; \preset{p} \; \cap \; \transitions_2 \\ \weights(t,p) > \weights(p,t)}} \ubound(t) \cdot \big(\weights(t,p) - \weights(p,t)\big)  $\label{line:pup}
}
        }
        \ForEach{$p \in \places$}{
        $\displaystyle \lbound(p) := M(p) - \sum_{\substack{t \in \transitions_2 \\
         \weights(p,t) > \weights(t,p)}} \ubound(t) \cdot \big(
          \weights(p,t) - \weights(t,p) \big)$\label{line:lb}
        }
        \Return{$\lbound, \ubound \models \varphi$};  \ \ \  *** See definition in Table~\ref{tab:lusat}
        \caption{$\mathit{reach}(N,M,\varphi)$: Overapproximation
         for checking if $\varphi$ can be satisfied by performing only player 2
         transitions, assuming that $\min \emptyset = \infty$ and $\sum \emptyset = 0$}%
    \label{alg:reachphi}
\end{algorithm}

\begin{table}[t]
\begin{align*}
        &\lbound,\ubound \models \true && \\
    &\lbound,\ubound \models t && \mbox{iff }
 \ubound(p) \geq \weights(p,t) \text{ for all } p \in \preset{t} \text{ and }
 \lbound(p) < \inhib(p,t) \text{ for all } p \in \inhibpreset{t}  \\
    &\lbound,\ubound \models \neg t && \mbox{iff }
 \lbound(p) < \weights(p,t) \text{ for some } p \in \preset{t} \text{ or }
 \ubound(p) \geq \inhib(p,t) \text{ for some } p \in \inhibpreset{t}  \\
    &\lbound,\ubound \models e_1 < e_2 && \mbox{iff } \lbound(e_1) < \ubound(e_2) \\
    &\lbound,\ubound \models e_1 \leq e_2 && \mbox{iff } \lbound(e_1) \leq \ubound(e_2) \\
    &\lbound,\ubound \models e_1 = e_2 && \mbox{iff } \max\{\lbound(e_1), \lbound(e_2)\} \leq
    \min\{\ubound(e_1),\ubound(e_2)\} \\
    &\lbound,\ubound \models e_1 \not= e_2 && \mbox{iff it is not the case that } \lbound(e_1)=\lbound(e_2)=\ubound(e_1)=\ubound(e_2) \\
    &\lbound,\ubound \models e_1 \geq e_2 && \mbox{iff } \ubound(e_1) \geq \lbound(e_2) \\
    &\lbound,\ubound \models e_1 > e_2 && \mbox{iff } \ubound(e_1) > \lbound(e_2) \\
    &\lbound,\ubound \models \deadlock &&\mbox{iff } \lbound,\ubound \not\models t \text{ for all } t \in \transitions \\
    &\lbound,\ubound \models \neg\deadlock &&\mbox{iff } \lbound,\ubound \models t \text{ for some } t \in \transitions \\
    &\lbound,\ubound \models \varphi_1 \land \varphi_2 &&\mbox{iff } \lbound,\ubound \models \varphi_1 \mbox{ and } \lbound,\ubound \models \varphi_2 \\
    &\lbound,\ubound \models \varphi_1 \lor \varphi_2 &&\mbox{iff } \lbound,\ubound \models \varphi_1 \mbox{ or } \lbound,\ubound \models \varphi_2 \\
\end{align*}

\begin{align*}
\lbound(c) & =c \text{ \ \ \ where $c$ is a constant} \\
\ubound(c) & =c \text{ \ \ \ where $c$ is a constant} \\
\lbound(e_1 + e_2) & = \lbound(e_1) + \lbound(e_2) \\
\ubound(e_1 + e_2) & = \ubound(e_1) + \ubound(e_2) \\
\lbound(e_1 - e_2) & = \lbound(e_1) - \ubound(e_2) \\
\ubound(e_1 - e_2) & = \ubound(e_1) - \lbound(e_2) \\
\lbound(e_1 * e_2) & = \min \{ \lbound(e_1)\cdot\lbound(e_2), \lbound(e_1)\cdot\ubound(e_2),
\ubound(e_1)\cdot\lbound(e_2), \ubound(e_1)\cdot\ubound(e_2)\} \\
\ubound(e_1 * e_2) & = \max \{ \lbound(e_1)\cdot\lbound(e_2), \lbound(e_1)\cdot\ubound(e_2),
\ubound(e_1)\cdot\lbound(e_2), \ubound(e_1)\cdot\ubound(e_2)\} \\
\end{align*}

\caption{Definition of $\lbound,\ubound \models \varphi$ assuming
that $\lbound(p)$ and $\ubound(p)$ are given for all $p \in P$}%
\label{tab:lusat}
\end{table}

\noindent
As a next step, we provide an algorithm that returns \emph{true} whenever there is a sequence
of player 2 actions that leads to a marking satisfying a given formula $\varphi$
(and hence overapproximates Condition \textbf{V} from the definition of a stable reduction).
The pseudocode is given in Algorithm~\ref{alg:reachphi}.
The algorithm uses an extended definition of formula satisfiability that,
instead of asking whether a formula holds in a given marking, specifies instead a range of markings
by two functions
$\lbound : \places \to \mathbb{N}^0$ for fixing a lower  bound on the number of tokens in places
and $\ubound : \places \to \mathbb{N}^0 \cup \{ \infty \}$ for specifying an upper bound.
A marking $M$ belongs to the range $\lbound, \ubound$ iff for all places $p \in P$ we
have $\lbound(p) \leq M(p) \leq \ubound(p)$. The extended satisfability predicate
$\lbound, \ubound \models \varphi$ is given in Table~\ref{tab:lusat}
and it must hold whenever there is a marking in the range specified by
$\lbound$ and $\ubound$ such that the marking satisfies the formula $\varphi$.
Finally, Algorithm~\ref{alg:reachphi} computes a safe overapproximation of the lower and upper bounds
such that
if $M \rtrans[w] M'$ for some $w \in \transitions_2^*$ then $\lbound(p) \leq M'(p) \leq \ubound(p)$ for
all $p \in P$.


\begin{lem}\label{lemma:reach2}
Let $N=\petrituple$ be a Petri net game, $M \in \markings$ a marking on $N$ and
$\varphi \in \Phi_N$ a formula.
If there is $w \in \ltslabels_2^*$ s.t. $M \rtrans[w] M'$ and $M' \models \varphi$
then $\mathit{reach}(N,M,\varphi)=\mathit{true}$.
\end{lem}
\begin{proof}
Algorithm~\ref{alg:reachphi} first computes for each place $p \in \places$
the upper bound $\ubound(p)$ and lower bound $\lbound(p)$
on the number of tokens that can appear in $p$
by performing any sequence of player 2 transitions, starting from the marking $M$.
The bounds are then used to return the value of the expression
$\lbound,\ubound \models \varphi$ that is defined in Table~\ref{tab:lusat}.

We shall first notice
if there is a marking $M'$ such that $\lbound(p) \leq M'(p) \leq \ubound(p)$
for all $p \in \places$ and $M' \models \varphi$ then $\lbound,\ubound \models \varphi$ holds.
This can be proved by a straightforward structural induction
on $\varphi$ while following the cases in Table~\ref{tab:lusat} where the
functions $\lbound$ and $\ubound$ are extended to arithmetical expressions used
in the query language such that for every marking $M'$ (as given above) and for every
arithmetical expressions $e$ we have $\lbound(e) \leq \eval[M'][e] \leq \ubound(e)$.

What remains to be established is the property that Algorithm~\ref{alg:reachphi}
correctly computes the lower and upper bounds for all places in the net. We do this
by proving the invariant for the repeat-until loop that claims that
for every $w \in \ltslabels_2^*$ such that $M \rtrans[w] M'$ we have
\begin{enumerate}
\item\label{inv1} $\displaystyle M(p) + \sum_{\substack{t \in w \\ \weights(t,p) > \weights(p,t)}}
       \big(\weights(t,p) - \weights(p,t)\big) \leq \ubound(p)$ for all $p \in \places$, and
\item\label{inv2} $\#_t(w) \leq \ubound(t)$ for all $t \in \transitions_2$ where
     $\#_t(w)$ denotes the number of occurences of the transition $t$ in the sequence $w$.
\end{enumerate}
Here the notation $t \in w$ means that a summand is added for every occurence of $t$ in the
sequence $w$.
We note that invariant (\ref{inv1}) clearly implies that
whenever $M \rtrans[w] M'$ for $w \in \ltslabels_2^*$ then
$M'(p) \leq \ubound(p)$ for all $p \in \places$.
Notice that the repeat-until loop in Algorithm~\ref{alg:reachphi} clearly terminates
since during the iteration of the loop $\ubound$ can only become smaller.

First, we notice that before entering the repeat-until loop, the invariant holds
because intitially the upper bound values are all set to $\infty$ and only
at line~\ref{line:cond1} the upper bound for a place $p$ is set
to $M(p)$ provided that the firing of any transition $t \in \transitions_2$
can never increase the number of tokens in $p$. This clearly satisfies
invariant (\ref{inv1}).

Let us now assume that both (\ref{inv1}) and (\ref{inv2}) hold at the beginning
of the execution of the repeat-until loop. Suppose that the value
$\ubound(t)$ is decreased for some transition $t$ by the assignment at line~\ref{line:tup}.
This means that there is a place $p \in \predecr{t}$ such that $\weights(p,t)>\weights(t,p)$,
meaning that firing of $t$ removes $\weights(t,p)-\weights(p,t)$ tokens from $p$.
As there can be at most $\ubound(p)$ tokens added to the place $p$
due to invariant (\ref{inv1}), this limits the number of times that the transition $t$
can fire to $\lfloor \frac{\ubound(p)}{\weights(t,p)-\weights(p,t)} \rfloor$
and hence it preserves invariant (\ref{inv2}).
Similarly, suppose that the value of $\ubound(p)$ is decreased for some place $p$ by the assignment
at line~\ref{line:pup}. Due to invariant (\ref{inv2}), we know that
every transition $t \in \transitions_2$ can be fired at most
$\ubound(t)$ times and hence adds at most $\ubound(t) \cdot \big(\weights(t,p) - \weights(p,t)\big)$
tokens to $p$.
As we add those contributions for all such transitions together with the number $M(p)$
of tokens in the starting marking $M$, we satisfy also invariant (\ref{inv1}).

Finally, the assignment at line~\ref{line:lb} provides a safe lower bound on the number
of tokens that can be in the place $p$, as due to invariant (\ref{inv2}) we know that
$\ubound(t)$ is the maximum number of times a transition $t$ can fire, and we subtract
the number of tokens that each $t$ removes from $p$ by $\ubound(t)$.
Hence, we can conclude that whenever $M \rtrans[w] M'$ for $w \in \ltslabels_2^*$ then
$\lbound(p) \leq M'(p) \leq \ubound(p)$ for all $p \in \places$ and the correctness
of the lemma is established.
\end{proof}

\begin{exa}\label{ex:algo-v}
In Figure~\ref{fig:algo-v-fig} we see a Petri net consisting of four places $\places = \{p_1, p_2, p_3,p_4\}$ and three player $2$ transitions $\transitions_2 = \{t_1, t_2, t_3\}$.
The weights are given as seen in the figure (arcs without any annotations have the default weight $1$)
and the initial marking contains three tokens in the place $p_1$.
Initially, for all $x \in \places \cup \transitions$ we have $\ubound(x) = \infty$ as seen in
line~\ref{line:init-inf} of Algorithm~\ref{alg:reachphi}.
In line~\ref{line:cond1} we can set the upper bound of some places if the number of tokens are non-increasing, i.e.\ for all $t \in \preset{p} \cap \transitions_2$ we have $\weights(p,t) \geq \weights(t,p)$.
In Figure~\ref{fig:algo-v-fig} this is the case only for $p_1$, we therefore have $\ubound(p_1) = M(p_1) = 3$.
Next, the upper bound for all places and transitions are calculated through a repeat-until loop.
The upper bound for transitions are found by checking, given the current upper bound on places, how many times we can fire a transition.
In line~\ref{line:tup} we get
\[\ubound(t_1) = \left\lfloor\frac{\ubound(p_1)}{\weights(p_1,t_1)-\weights(t_1,p_1)}\right\rfloor =
\left\lfloor\frac{3}{1-0}\right\rfloor = 3\] and
\begin{align*}
\ubound(t_2)
    &= \min \left\{\left\lfloor\frac{\ubound(p_1)}{\weights(p_1,t_2)-\weights(t_2,p_1)}\right\rfloor, \left\lfloor\frac{\ubound(p_3)}{\weights(p_3,t_2)-\weights(t_2,p_3)}\right\rfloor\right\} \\
    &= \min \left\{ \left\lfloor\frac{3}{2-0}\right\rfloor, \left\lfloor\frac{\infty}{1-0} \right\rfloor\right\} \\
    &= \min\{1,\infty\} = 1.
\end{align*}
In the next iteration, at line~\ref{line:pup} we get
\[
    \ubound(p_2) = M(p_2) + \ubound(t_1) \cdot (\weights(t_1,p_2) - \weights(p_2,t_1)) = 0 + 3 \cdot (1 - 0) = 3\] and similarly
\[\ubound(p_4) = M(p_4) + \ubound(t_2) \cdot (\weights(t_2,p_4) - \weights(p_4,t_2)) = 0 + 1 \cdot (1 - 0) = 1.\]
Afterwards, there are no further changes to be made to the upper bounds and the repeat-until loop terminates.
Finally, the calculated lower bounds for all places are $0$ in our example.
\end{exa}

\begin{figure}[t]
    	\centering
    	\begin{tikzpicture}[font=\scriptsize,xscale=2.4,yscale=1.8]
		\tikzstyle{arc}=[->,>=stealth,thick]
		\tikzstyle{every place}=[minimum size=6mm,thick]
		\tikzstyle{every transition}=[fill=black,minimum width=2mm,minimum height=5mm]
		\tikzstyle{oppTransition}=[draw=black,fill=white,minimum width=2mm,minimum height=5mm,thick]
		\tikzstyle{token}=[fill=white,text=black]

		\begin{scope}
			\node [place,label=$p_1$] at (0,0.6) (p1) {$\bullet \bullet \bullet$};

			\node [oppTransition,label=$t_1$] at (1,0.6) (t1) {};

			\node [oppTransition,label=$t_2$] at (1,-0.6) (t2) {};

			\node [place,label=$p_2$] at (2,0.6) (p2) {};

			\node [place,label=$p_3$] at (0,-0.6) (p3) {};

			\node [oppTransition,label=$t_3$] at (-1,-0.6) (t3) {};

			\node [place,label=$p_4$] at (2,-0.6) (p4) {};

			\draw [->,>=stealth,thick] (p1) to (t1);
			\draw [->,>=stealth,thick] (p1) to node[pos=0.5,above]{$2$} (t2);
			\draw [->,>=stealth,thick] (t1) to (p2);
			\draw [->,>=stealth,thick] (p3) to (t2);
			\draw [->,>=stealth,thick] (t3) to (p3);
			\draw [->,>=stealth,thick] (t2) to (p4);
		\end{scope}
    	\end{tikzpicture}
    \caption{Example Petri Net for Algorithm~\ref{alg:reachphi}}%
    \label{fig:algo-v-fig}
\end{figure}
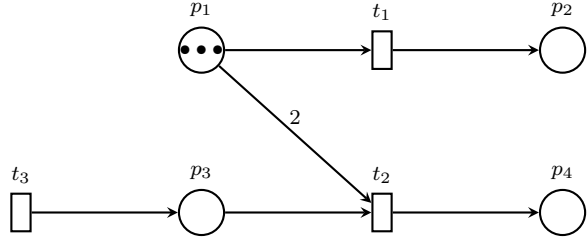

Before we can state our main theorem, we need to find an overapproximation method for determining
safe transitions. This can be done by analysing the increasing presets and postsets of transitions
as demonstrated in the following lemma.

\begin{lem}[Safe Transition]\label{lemma:trans-safe}
        Let $N = \petrituple$ be a \PNG
and $t \in \transitions$ a transition.
If $\postincr{t} \cap \preset{T_2} = \emptyset$ and $\predecr{t} \cap \inhibpreset{T_2} = \emptyset$ then $t$ is safe in any marking of $N$.
\end{lem}
\begin{proof}
        Assume $\postincr{t} \cap \preset{T_2} = \emptyset$ and $\predecr{t} \cap \inhibpreset{T_2} = \emptyset$.
        We prove directly that $t$ is safe in $M$.
        Let     $w \in {(\transitions_1 \setminus \{ t \})}^*$ s.t. $M \rtrans[w] M'$, $\en_2(M') = \emptyset$, and $M \rtrans[tw] M''$.
        The only difference between $M'$ and $M''$ is that $t$ is fired first and we have $M''(p') = M'(p') + \weights(t,p') - \weights(p',t)$ for all $p \in \places$.
        Then for all $t' \in \transitions_2$ we have that there either exists $p \in \preset{t'}$ s.t. $M'(p) <\weights(p,t')$, or there exists $p' \in \inhibpreset{t'}$ s.t. $M'(p') \geq \inhib(p',t')$.
        In the first case, since $\postincr{t} \cap \preset{T_2} = \emptyset$, we must have $\weights(t,p) \leq \weights(p,t)$ which implies $M''(p) \leq M'(p)$ and $t' \notin \en(M'')$.
        In the second case, since $\predecr{t} \cap \inhibpreset{T_2} = \emptyset$, we must have $\weights(t,p') \geq \weights(p',t)$, which implies $M''(p') \geq M'(p')$ and $t' \notin \en(M'')$.
        Therefore $t$ is safe in $M$.
\end{proof}

We can now provide a list of syntactic conditions
that guarantee the stability of a given reduction and state the main theorem of this section.

\begin{thm}[Stable Reduction Preserving Closure]\label{prop:turn-game-preservation}
Let $N = \petrituple$ be a \PNG, $\varphi$ a formula, and \reduction a reduction of $\lts(N)$ such that for all $M \in \markings$ the following conditions hold.
\begin{enumerate}
    \item\label{item:contested} If $\en_1(M) \neq \emptyset$ and $\en_2(M) \neq \emptyset$ then $\en(M) \subseteq \reduction(M)$.
    \item\label{item:safe} If $\en_1(M) \cap \reduction(M) \nsubseteq \safe{M}$ then $\en_1(M) \subseteq \reduction(M)$.
    \item\label{item:interesting} $\interest{M}{\varphi} \subseteq \reduction(M)$
    \item\label{item:opponent-1} If $\en_1(M) = \emptyset$ then $\transitions_1 \subseteq \reduction(M)$.
    \item\label{item:opponent-2} If $\en_2(M) = \emptyset$ then $\transitions_2 \subseteq \reduction(M)$.
    \item\label{item:prop-disabled} For all $t \in \reduction(M)$ if $t \notin en(M)$ then either
	\begin{enumerate}
		\item\label{item:disabled-1} there exists $p \in \preset{t}$ s.t. $M(p) < \weights(p,t)$ and $\preincr{p} \subseteq \reduction(s)$, or
		\item\label{item:disabled-2} there exists $p \in \inhibpreset{t}$ s.t. $M(p) \geq \inhib(p,t)$ and $\postdecr{p} \subseteq \reduction(s)$.
	\end{enumerate}
    \item\label{item:prop-enabled} For all $t \in \reduction(M)$ if $t \in \en(M)$ then
    \begin{enumerate}
		\item\label{item:enabled-1} for all $p \in \predecr{t}$ we have $\postset{p} \subseteq \reduction(M)$, and
		\item\label{item:enabled-2} for all $p \in \postincr{t}$ we have $\inhibpostset{p} \subseteq \reduction(M)$.
	\end{enumerate}
	\item\label{item:deadlock} If $\en_2(M) \neq \emptyset$ then there exists $t \in \en_2(M) \cap \reduction(M)$ s.t. $\postdecr{(\preset{t})} \cup \preincr{(\inhibpreset{t})} \subseteq \reduction(M)$.
    \item\label{item:visible} If $\en_1(M) = \emptyset$ and $\mathit{reach}(N,M,\varphi)=\mathit{true}$ then $\en(M) \subseteq \reduction(M)$.
\end{enumerate}
Then \reduction satisfies~\ref{rule:stub-init},~\ref{rule:stub},~\ref{rule:reach},~\ref{rule:game-1},~\ref{rule:game-2},~\ref{rule:safe},~\ref{rule:visible} and~\ref{rule:dead}.
\end{thm}
\begin{proof}
	We shall argue that any reduction $\reduction$ satisfying the conditions of the theorem also satisfies the~\ref{rule:stub-init},~\ref{rule:stub},~\ref{rule:reach},~\ref{rule:game-1},~\ref{rule:game-2},~\ref{rule:safe},~\ref{rule:visible}, and~\ref{rule:dead} conditions.
	\begin{itemize}[leftmargin=10mm]
		\item[(\ref{rule:stub-init})]
		Follows from Condition~\ref{item:contested}.
		\item[(\ref{rule:stub})]
		Let $M,M' \in \markings$ be markings, $t \in \reduction(M)$, and $w \in \overline{\reduction(M)}^*$.
		We will show that if $M \rtrans[wt] M'$ then $M \rtrans[tw] M'$.
    		Let $M_w \in \mathcal{M}(N)$ be a marking s.t. $M \rtrans[w] M_w$.
		Assume for the sake of contradiction that $t \notin \en(M)$.
		As $t$ is disabled in $M$, there must be $p \in \preset{t}$ such that $M(p) < \weights(p,t)$ or there is $p \in \inhibpreset{t}$ such that $M(p) \geq \inhib(p,t)$.
		In the first case, due to Condition~\ref{item:disabled-1} all the transitions that can add tokens to $p$ are included in $\reduction(M)$.
		Since $w \in \overline{\reduction(M)}^*$ this implies that $M_w(p) < \weights(p,t)$ and $t \notin \en(M_w)$ contradicting our assumption that $M_w \rtrans[t] M'$.
		In the second case, due to Condition~\ref{item:disabled-2} all the transitions that can remove tokens from $p$ are included in $\reduction(M)$.
		Since $w \in \overline{\reduction(M)}^*$ this implies that $M_w(p) \geq \inhib(p,t)$ and $t \notin \en(M_w)$ contradicting our assumption that $M_w \rtrans[t] M'$.
		Therefore we must have that $t \in \en(M)$.

		Since $t \in \en(M)$ there is $M_t \in \markings$ s.t. $M \rtrans[t] M_t$.
		We have to show that $M_t \rtrans[w] M'$ is possible.
 		For the sake of contradiction, assume that this is not the case.
		Then there must exist a transition $t'$ that occurs in $w$ that became disabled because $t$ was fired.
		There are two cases: $t$ removed one or more tokens from a shared pre-place $p \in \predecr{t} \cap \preset{t'}$ or added one or more tokens to a place $p \in \postincr{t} \cap \inhibpreset{t'}$.
		In the first case, due to Condition~\ref{item:enabled-1} all the transitions that can remove tokens from $p$ are included in $\reduction(M)$, implying that $t' \in \reduction(M)$.
		Since $w \in \overline{\reduction(M)}^*$ such a $t'$ cannot exist.
		In the second case, due to Condition~\ref{item:enabled-2} all the transitions that can add tokens to $p$ are included in $\reduction(M)$, implying that $t' \in \reduction(M)$.
		Since $w \in \overline{\reduction(M)}^*$ such a $t'$ cannot exist.
		Therefore we must have that $M_t \rtrans[w] M'$ and we can conclude with $M \rtrans[tw] M'$.
		\item[(\ref{rule:reach})]
		Follows from Condition~\ref{item:interesting} and Lemma~\ref{lemma:reach}.
		\item[(\ref{rule:game-1})]
		Let $M \in \markings$ be a marking and $w \in \overline{\reduction(M)}^*$ s.t. $M \rtrans[w] M'$.
		We will show that if $\en_2(M) = \emptyset$ then $\en_2(M') = \emptyset$.
		Assume that $\en_2(M) = \emptyset$.
		Then by Condition~\ref{item:opponent-2} we have $\transitions_2 \subseteq \reduction(M)$.
		Let $t \in \transitions_2$ be a player $2$ transition.
		By Condition~\ref{item:prop-disabled} we know that either there exists $p \in \preset{t}$ s.t. $M(p) < \weights(p,t)$ and $\preincr{p} \subseteq \reduction(s)$, or there exists $p \in \inhibpreset{t}$ s.t. $M(p) \geq \inhib(p,t)$ and $\postdecr{p} \subseteq \reduction(s)$.
		In the first case, in order to enable $t$ at least one transition from $\preincr{p}$ has to be fired.
		However, we know $\preincr{p} \subseteq \reduction(s)$ is true, and therefore none of the transitions in $\preincr{p}$ can occur in $w$, which implies $t \notin \en_2(M')$.
		In the second case, in order to enable $t$ at least one transition from $\postdecr{p}$ has to be fired.
		However, we know $\postdecr{p} \subseteq \reduction(s)$ is true, and therefore none of the transitions in $\postdecr{p}$ can occur in $w$, which implies $t \notin \en_2(M')$.
		These two cases together imply that $\en_2(M') = \emptyset$.
		\item[(\ref{rule:game-2})]
		Follows the same approach as~\ref{rule:game-1}.
		\item[(\ref{rule:safe})]
		Follows from Condition~\ref{item:safe}.
        \item[(\ref{rule:visible})]
        Follows from Condition~\ref{item:visible} and Lemma~\ref{lemma:reach2}.
        Notice that if $\en_1(M) \neq \emptyset$ then the antecedent of Condition~\ref{rule:visible} never holds if $\en_2(M) = \emptyset$ unless $M$ is already a goal marking, or $M$ is a mixed state and the consequent of Condition~\ref{rule:visible} always holds due to Condition~\ref{rule:stub-init}.
		\item[(\ref{rule:dead})]
		Let $M \in \markings$ be a marking and $w \in \overline{\reduction(M)}^*$ s.t. $M \rtrans[w] M'$.
		We will show that if $\en_2(M) \neq \emptyset$ then there exists $t \in \en_2(M) \cap \reduction(M)$ s.t. $t \in \en_2(M')$.
		Assume that $\en_2(M) \neq \emptyset$.
		From Condition~\ref{item:deadlock} we know that there exists $t \in \en_2(M) \cap \reduction(M)$ s.t. $\postdecr{(\preset{t})} \cup \preincr{(\inhibpreset{t})} \subseteq \reduction(M)$.
		Assume for the sake of contradiction that $t \notin \en_2(M')$.
		In this case there must either exist $p \in \preset{t}$ s.t. $M'(p) < \weights(p,t)$, or there exists $p \in \inhibpreset{t}$ s.t. $M'(p) \geq \inhib(p,t)$.
		In the first case, since $t \in \en_2(M)$ we have that $M(p) \geq  \weights(p,t)$.
		Therefore at least one transition from $\postdecr{p}$ has to have been fired.
		However, we know $\postdecr{(\preset{t})} \subseteq \reduction(M)$ is true, and therefore none of the transitions in $\postdecr{p}$ can occur in $w$, which implies $M'(p) \geq  \weights(p,t)$, a contradiction.
		In the second case, since $t \in \en_2(M)$ we have that $M(p) < \inhib(p,t)$.
		Therefore at least one transition from $\preincr{p}$ has to have been fired.
		However, we know $\preincr{(\preset{t})} \subseteq \reduction(M)$ is true, and therefore none of the transitions in $\preincr{p}$ can occur in $w$, which implies $M'(p) < \inhib(p,t)$, a contradiction.
		Therefore $t \notin \en_2(M')$ cannot be true, and we must have that $t \in \en_2(M')$.
	\end{itemize}
	This completes the proof of the theorem.
\end{proof}

In Algorithm~\ref{alg:saturation} we provide a pseudocode
for calculating stubborn sets for a given marking.
It essentially rephrases Theorem~\ref{prop:turn-game-preservation} into an executable code.
The algorithm calls Algorithm~\ref{alg:saturation-2}
that saturates a given set to satisfy
Conditions~\ref{item:prop-disabled} and~\ref{item:prop-enabled} of
Theorem~\ref{prop:turn-game-preservation}.

\begin{algorithm}[t]
	\SetKwInOut{Input}{input}
	\SetKwInOut{Output}{output}
	\Input{A Petri net game $N=\petrituple$ and $M \in \markings$
        and formula $\varphi$}
	\Output{$X \subseteq \transitions$ where $X$ is a stable
          stubborn set for $M$}

	\If{$\en(M) = \emptyset$}{
		\Return{$\transitions$};
	}
	\If{$\en_1(M) \neq \emptyset \land \en_2(M) \neq \emptyset$}{
		\Return{$\transitions$};\label{line:mixed}
	}

	$Y$ := $\emptyset$; 

	\eIf{$\en_1(M) = \emptyset$}{
		\If{$\mathit{reach}(N,M,\varphi)$}{
			\Return{$\transitions$};\label{line:visible}
		}

		Pick any $t \in \en_2(M)$;\label{line:dead-pick}

		$Y$ := $\transitions_1 \cup t \cup \postdecr{(\preset{t})} \cup \preincr{(\inhibpreset{t})}$;\label{line:p1} 

	}{
		$Y$ := $\transitions_2$;\label{line:p2}
	}

	$Y$ := $Y \cup \interest{M}{\varphi}$;\label{line:interesting}

	$X$ := $\mathit{Saturate}(Y)$;

	\If{$X \cap \en_1(M) \nsubseteq \safe{M}$}{
		\Return{$\transitions$};\label{line:safe}
	}

	\Return{$X$};\label{line:terminate}
	\caption{Computation of $\reduction(M)$ for some
     stable reduction $\reduction$}%
    \label{alg:saturation}
\end{algorithm}

\begin{algorithm}
	$X$ := $\emptyset$; 

	\While{$Y \neq \emptyset$} {

		Pick any $t \in Y$;

		\eIf{$t \notin \en(M)$}{
			\eIf{$\exists p \in \preset{t}.\ M(p) < \weights(p,t)$}{
				Pick any $p \in \preset{t}$ s.t. $M(p) < \weights(p,t)$;\label{line:pick-dis-1}

				$Y$ := $Y \cup (\preincr{p} \setminus X$);\label{line:dis-1-add} 
			}{
				Pick any $p \in \inhibpreset{t}$ s.t. $M(p) \geq \inhib(p,t)$;\label{line:pick-dis-2}

				$Y$ := $Y \cup (\postdecr{p} \setminus X$);\label{line:dis-2-add} 
			}
		}{
			$Y$ := $Y \cup ((\postset{(\predecr{t})} \cup \inhibpostset{(\postincr{t})}) \setminus X)$;\label{line:enabled} 
		}

		$X$ := $X \cup \{t\}$;\label{line:retire} 

		$Y$ := $Y \setminus \{t\}$; 
	}
	\Return{$X$};\label{line:final-return}

	\caption{$\mathit{Saturate}(Y)$}%
    \label{alg:saturation-2}
\end{algorithm}

\begin{thm}\label{thm:alg}
	Algorithm~\ref{alg:saturation}
terminates and returns $\reduction(M)$
for some stable reduction $\reduction$.
\end{thm}
\begin{proof}
	\textit{Termination}. If $\en_1(M) \neq \emptyset$ and $\en_2(M) \neq \emptyset$ then we terminate in line~\ref{line:mixed}.
	Otherwise $Y \neq \emptyset$ and we enter the while-loop in Algorithm~\ref{alg:saturation-2}.
	Notice that $X \cap Y = \emptyset$ is always the case in the execution of Algorithm~\ref{alg:saturation-2}.
	We never remove transitions from $X$ after they have been added.
	Therefore, since in line~\ref{line:retire} of Algorithm~\ref{alg:saturation-2} a new transition is added to $X$ at the end of each loop iteration, the loop can iterate at most once for each transition.
	Since $T$ is finite by the Petri Net Game definition, the loop iterates a finite number of times, and Algorithm~\ref{alg:saturation-2} terminates.
	If $\en_1(M) \cap X \nsubseteq \safe{M}$ then we terminate in line~\ref{line:safe} of Algorithm~\ref{alg:saturation}, and otherwise we return in line~\ref{line:terminate} and Algorithm~\ref{alg:saturation} terminates.

	\textit{Correctness}. It was shown that the construction in Theorem~\ref{prop:turn-game-preservation} results in a set that is a stubborn set of a stable reduction.
	 It is therefore sufficient to show that Algorithm~\ref{alg:saturation} replicates the construction.
	 Notice that every transition that is added to $Y$ is eventually added to $X$ in line~\ref{line:retire} and returned in line~\ref{line:final-return} of Algorithm~\ref{alg:saturation-2}.
	 Let $t \in Y$ and we discuss that all conditions of Theorem~\ref{prop:turn-game-preservation} hold upon termination.
	\begin{itemize}
\item
	 Condition~\ref{item:contested}: If $\en_1(M) \neq \emptyset$ and $\en_2(M) \neq \emptyset$ then we return $\transitions$ in line~\ref{line:mixed} of Algorithm~\ref{alg:saturation}.
\item
	 Condition~\ref{item:safe}: If $\en_1(M) \cap \reduction(M) \nsubseteq \safe{M}$ then we return $\transitions$ in line~\ref{line:safe} of Algorithm~\ref{alg:saturation}.
\item
	 Condition~\ref{item:interesting}: We have $\interest{M}{\varphi} \subseteq Y$ in line~\ref{line:interesting} of Algorithm~\ref{alg:saturation}.
\item
	 Condition~\ref{item:opponent-1}: We have $\transitions_1 \subseteq Y$ in line~\ref{line:p1} of Algorithm~\ref{alg:saturation}.
\item
	 Condition~\ref{item:opponent-2}: We have $\transitions_2 \subseteq Y$ in line~\ref{line:p2} of Algorithm~\ref{alg:saturation}.
\item
	 Condition~\ref{item:disabled-1}: In line~\ref{line:pick-dis-1} we pick any $p \in \preset{t}$ s.t. $M(p) < \weights(p,t)$, and in line~\ref{line:dis-1-add} of Algorithm~\ref{alg:saturation-2} we add $\preincr{p}$ to $Y$.
\item
	 Condition~\ref{item:disabled-2}: In line~\ref{line:pick-dis-2} we pick any $p \in \inhibpreset{t}$ s.t. $M(p) \geq \inhib(p,t)$, and in line~\ref{line:dis-2-add} of Algorithm~\ref{alg:saturation-2} we add $\postdecr{p}$ to $Y$.
\item
	 Condition~\ref{item:enabled-1}: In line~\ref{line:enabled} of Algorithm~\ref{alg:saturation-2} we add $\postset{(\predecr{t})}$ to $Y$.
\item
	 Condition~\ref{item:enabled-2}: In line~\ref{line:enabled} of Algorithm~\ref{alg:saturation-2} we add $\inhibpostset{(\postincr{t})}$ to $Y$.
\item
	 Condition~\ref{item:deadlock}: In line~\ref{line:dead-pick} of Algorithm~\ref{alg:saturation} we pick any $t' \in \en_2(M)$ and in line~\ref{line:p1} we add $\postdecr{(\preset{t})} \cup \preincr{(\inhibpreset{t})}$ to $Y$.
\item
     Condition~\ref{item:visible}: If $\en_1(M) = \emptyset$ and $\mathit{reach}(N,M,\varphi)=\mathit{true}$ then we return $T$ at line~\ref{line:visible} of Algorithm~\ref{alg:saturation}.
     \qedhere
\end{itemize}
\end{proof}

\begin{rem}
In the actual implementation of the algorithm, we first saturate
only over the set of interesting transitions and in the case that
$\mathit{Saturate}(\interest{M}{\varphi}) \cap \en(M) = \emptyset$,
we do not explore  any of the successors of the marking $M$
as we know that no goal marking can be reached from $M$ (this
follows from Lemma~\ref{lemma:early-termination}).
\end{rem}

%% file: sections/implementation.tex
\section{Implementation and Experiments}
We extend the Petri net verification engine
\texttt{verifypn}~\cite{jensen2016tapaal}, a part
of the TAPAAL tool suite~\cite{david2012tapaal}, to experimentally
demonstrate the viability of our approach.
The synthesis algorithm for solving Petri net games
is an adaptation of the dependency graph fixed-point computation
from~\cite{jensen2018discrete,jensen2016real} that we reimplement
in \texttt{C++} while utilising PTries~\cite{JLS:ICTAC:17}
for efficient state storage. The source code is available
under GPLv3~\cite{repeatability}.
We conduct a series of experiments using the following scalable case studies.
\begin{itemize}
\item 
  In \emph{Autonomous Intersection Management} (AIM) vehicles move
  at different speeds towards an intersection and we want to ensure the
  absence of collisions. We model the problem as a Petri net
  game 
  and refer to each instance as AIM-$W$-$X$-$Y$-$Z$ where
  $W$ is the number of intersections with lanes of length $X$, $Z$ is the number
  of cars, and $Y$ is the number of different speeds for each car.
  The controller assign speeds to cars while the
  environment aims to cause a collision. The goal marking is where
  all cars reach their destinations while there are no collisions.
\item We reformulate the classical \emph{Producer Consumer System} (PCS)
  as a Petri net game. In each instance PCS-$N$-$K$ the
  total of $N$ consumers
 (controlled by the environment) and $N$ producers
 (controlled by the controller) share $N$ buffers.
 Each consumer and producer has a fixed buffer
 to consume/produce from/to, and each consumer/producer has $K$
 different randomly chosen consumption/production rates.
 The game alternates in rounds where the players choose for each
 consumer/producer appropriate buffers and rates.
 The goal of the game is to ensure that the consumers have always
 enough products in the selected buffers while at the same time the
 buffers have limited capacity and may not overflow.
\item The \emph{Railway Scheduling Problem} contains four instances
 modeling the Danish train station Lyngby and three of its smaller
 variants. The scheduling problem, including the station layout,
 was originally described as a game
 in~\cite{kasting2016synthesis} and each instance is annotated by
 a number
$N$ representing the number of trains that migrate through the railway
 network. The controller controls the lights and switches, while
 the environment moves the trains. The goal of the controller is to
 make sure that all trains reach (without any collisions)
 their final destinations.
\item The \emph{Nim} (NIM-$K$-$S$) Petri net game was described
 in~\cite{T:CIMCA:08} as a two player game where the players
 in rounds repeatedly remove between $1$ and $K$ pebbles from an initial stack
 containing $S$ pebbles.
 The player that has a turn and an empty stack
 of pebbles loses. In our (equivalent) model,
 we are instead adding pebbles to an initially empty stack and the player
 that first adds to or above the given number $S$ loses.
\item The \emph{Manufacturing Workflow} (MW) contains
 instances of a software product line Petri net model presented
 in~\cite{QUINTANILLA2013342}. The net describes a series of possible
 ways of configuring a product (performed by the environment)
 while the controller aims to construct a requested product.
 The model instance MW-$N$ contains $N$ possible choices of product
 features.
\item The \emph{Order Workflow} (OW) Petri net game model is taken
 from~\cite{10.1007/978-3-540-30468-511} and the goal of the game
 is to synthesise a strategy that guarantees workflow soundness,
 irrelevant of the choices made by the environment.
 We scale the workflow by repeatedly re-initialising the workflow
 $N$ times (denoted by OW-$N$).
\item In \emph{Flexible Manufacturing Systems} (FMS) we use the Petri net models
 from~\cite{LZ:04,AE:98} modeling different production lines with shared resources.
 The Petri nets FMS-D~\cite{AE:98} and FMS-C~\cite{LZ:04} both contain a
 deadlock and the problem is to control a small subset of transitions
 so that the deadlock can be avoided. The models are scaled by the number
 of resources and products in the line. The goal in the FMS-N~\cite{LZ:04} model is to
 control a subset of transitions in the net in order to guarantee that a given
 resource (Petri net place) never becomes empty.
\end{itemize}

\noindent
All experimental evaluation is run on AMD Epyc 7551  Processors
with 110 GB memory limitation and 12 hours timeout (we measure only
the execution time without the parsing time of the models).
We use for all experiments the depth first search strategy and
we only report the examples where the algorithms both with and without
partial order reduction returned a result within the time and memory limits.
We provide a reproducibility
package with all models and experimental data~\cite{repeatability}.

\input{sections/results}


%% file: sections/results.tex
\subsection*{Results}
\input{sections/resulttable}
Table~\ref{tab:results} shows the experimental evaluation, displaying the
relative gain in computation time (in seconds)  without (NORMAL) and with (POR) partial order
reduction as well in the number of unique markings (in thousands) that were stored during the
fixed-point computation on the constructed dependency graph.
The results demonstrate significant reductions across all models, in some
cases like in NIM and MW even of several degrees of magnitude due to the exponential
speed up when using partial order reduction.
The case studies FMS-N and FMS-C show a large and consistent reduction in time across
all instance sizes.
Other models like AIM, PCS, OW and FMS-D show a moderate but significant reduction.
We observe that the time reduction is generally only few percent
different from the reduction in the number of explored markings,
indicating only a few percent overhead for computing (on-the-fly)
the stubborn sets. In the FMS-C and in particular the FMS-N model we can see that
we achieve significantly larger reduction in running time than in the reduced number
of stored markings. This is caused by the fact that the partial order reduction
reduces also the number of possible paths in which a certain marking can be discovered.

Our partial order technique noticably speeds up the computation in Lyngby2 model, but there
are also two instances of the LyngbySmall models where the reduction both in time
and size of the state space is less significant.
We conjecture that this is because the search strategy changes when partial order
reduction is applied and this results in the fact that we have to search in these two instances
a larger portion of the generated dependency graph before we obtain a conclusive answer.
Nevertheless, in general the experiments confirm the  high practical applicability
of partial order reduction for 2-player games with only minimal overhead
for computing the stubborn sets.
exponential

%% file: sections/resulttable.tex
\begin{table}[hbtp]
\centering
\begin{tabular}{lrrrrrr}

& \multicolumn{2}{c}{Time (seconds)} &\multicolumn{2}{c}{Markings $\times1000$} & \multicolumn{2}{c}{Reduction}  \\
Model& NORMAL & POR & NORMAL & POR& \%Time & \%Markings \\\toprule

AIM-13-100-6-11&54.15&19.86&1702&510&63&70\\ 
AIM-13-100-6-16&76.07&28.71&2464&740&62&70\\ 
AIM-13-150-9-16&162.10&115.30&3696&2455&29&34\\ 
AIM-13-150-9-21&212.80&153.00&4853&3331&28&31\\ 
AIM-14-150-9-16&200.30&142.90&4259&2865&29&33\\ 
AIM-15-150-9-16&243.30&172.50&4861&3205&29&34\\ 
\midrule

PCS-2-3&49.71&37.86&13660&9839&24&28\\ 
PCS-2-4&181.50&126.60&37580&25625&30&32\\ 
PCS-2-5&488.40&331.90&84059&55049&32&35\\ 
PCS-2-6&1226.00&756.50&164096&104185&38&37\\ 
\midrule

LyngbySmall2&1.47&0.01&359&2&99&99\\ 
LyngbySmall3&9.79&5.59&2118&1165&43&45\\ 
LyngbySmall4&60.66&45.32&11605&7407&25&36\\ 
Lyngby2&1440.00&116.00&137169&11213&92&92\\ 
\midrule

NIM-5-49500&3.88&1.24&1635&595&68&64\\ 
NIM-7-49500&14.40&1.78&5282&753&88&86\\ 
NIM-9-49500&65.36&2.42&17326&963&96&94\\ 
NIM-11-49500&326.40&3.28&59491&1167&99&98\\ 
\midrule

MW-20&18.03&0.02&4333&4&100&100\\ 
MW-30&93.71&0.04&14643&6&100&100\\ 
MW-40&311.50&0.06&34733&9&100&100\\ 
MW-50&795.30&0.09&67869&11&100&100\\ 
MW-60&1749.00&0.13&117313&13&100&100\\ 
\midrule

OW-100000&4.20&3.25&2300&1800&23&22\\ 
OW-1000000&52.04&40.37&23000&18000&22&22\\ 
OW-10000000&591.30&435.60&230000&180000&26&22\\ 
\midrule

FMS-D-4&40.81&35.38&7145&6682&13&6\\ 
FMS-D-5&98.74&90.98&15735&15156&8&4\\ 
FMS-D-6&170.20&159.90&25873&25265&6&2\\ 
FMS-D-7&246.80&238.20&36569&35918&3&2\\ 
\midrule

FMS-C-300&218.50&105.40&24730&16103&52&35\\ 
FMS-C-400&349.50&170.40&32877&21411&51&35\\ 
FMS-C-500&471.30&239.90&41026&26718&49&35\\ 
FMS-C-600&587.10&285.60&49173&32024&51&35\\ 
\midrule

FMS-N-9000&52.84&16.69&10423&8579&68&18\\ 
FMS-N-29000&201.90&64.08&33583&27639&68&18\\ 
FMS-N-49000&372.50&119.00&56743&46699&68&18\\ 
FMS-N-69000&538.70&179.70&79903&65759&67&18\\ 
\end{tabular}
\caption{Experiments with and without partial order reduction (POR and NORMAL)}%
\label{tab:results}
\end{table}

%% file: sections/conclusion.tex
\section{Conclusion}
We generalised the partial order reduction technique based on
stubborn sets from plain reachability to a game theoretical
setting. This required a nontrivial extension of the classical
conditions on stubborn sets so
that a state space reduction can be achieved for both players
in the game. In particular, the computation of the stubborn sets
for player 2 (uncontrollable transitions) needed a new technique for interval approximation on the number
of tokens in reachable markings.
We proved the correctness of our approach and instantiated it to the case
of Petri net games. We provided (to the best of our knowledge) the first
implementation of partial order reduction for Petri net games
and made it available as a part of the model checker TAPAAL\@.
The experiments show promising results on a number of case studies,
achieving in general a substantial state space reduction with only a
small overhead for computing the stubborn sets. In the future work,
we plan to combine our contribution with a recent insight on how to
effectively use partial order reduction in the timed
setting~\cite{boenneland2018start} in order to extend our framework
to general timed games.

\textbf{Acknowledgments.}
We are grateful to Thomas Neele from Eindhoven University of Technology for
letting us know about the false claim in Lemma~\ref{lemma2} that was
presented in the conference version of this article.
The counterexample, presented in Remark~\ref{rem:error}, is attributed to him.
We are obliged to Antti Valmari for noticing that condition {\bf C}
in our conference paper is redundant and
can be substituted by conditions {\bf W} and {\bf D}, as it is done in this article.
We also thank the anonymous reviewers for their numerous suggestions that helped us
to improve the quality of the presentation.
The research leading to these results has received funding from the
project \mbox{DiCyPS} funded by the Innovation Fund Denmark,
the ERC Advanced Grant LASSO and DFF project QASNET\@.

%% file: main.bbl
\newcommand{\etalchar}[1]{$^{#1}$}
\begin{thebibliography}{QKCC13}

\bibitem[ABDL16]{alechina2016complexity}
N.~Alechina, N.~Bulling, S.~Demri, and B.~Logan.
\newblock {On the Complexity of Resource-Bounded Logics}.
\newblock In {\em Reachability Problems}, volume 9899 of {\em LNCS}, pages
  36--50. Springer Berlin Heidelberg, 2016.

\bibitem[AE98]{AE:98}
I.~B. Abdallah and H.~A. ElMaraghy.
\newblock {Deadlock Prevention and Avoidance in FMS: A Petri Net Based
  Approach}.
\newblock {\em The International Journal of Advanced Manufacturing Technology},
  14(10):704--715, 1998.
\newblock Springer.

\bibitem[BHSS12]{berard2012concurrent}
B.~Bérard, S.~Haddad, M.~Sassolas, and N.~Sznajder.
\newblock {Concurrent Games on VASS with Inhibition}.
\newblock In {\em International Conference on Concurrency Theory}, volume 7454
  of {\em LNCS}, pages 39--52. Springer Berlin Heidelberg, 2012.

\bibitem[BJL{\etalchar{+}}18]{boenneland2018start}
F.M. Bønneland, P.G. Jensen, K.G. Larsen, M.~Muñiz, and J.~Srba.
\newblock {Start Pruning When Time Gets Urgent: Partial Order Reduction for
  Timed Systems}.
\newblock In {\em Computer Aided Verification}, volume 10981 of {\em LNCS},
  pages 527--546. Springer-Verlag, 2018.

\bibitem[BJL{\etalchar{+}}19]{boenneland2019partial}
F.M. Bønneland, P.G. Jensen, K.G. Larsen, M.~Mu{\~{n}}iz, and J.~Srba.
\newblock {Partial Order Reduction for Reachability Games}.
\newblock In {\em International Conference on Concurrency Theory}, volume 140
  of {\em Leibniz International Proceedings in Informatics}, pages 23:1--23:15.
  Schloss Dagstuhl--Leibniz-Zentrum fuer Informatik, 2019.

\bibitem[BJL{\etalchar{+}}20]{repeatability}
Frederik~Meyer Bønneland, Peter~Gjøl Jensen, Kim~Guldstrand Larsen, Marco
  Mũniz, and Jiri Srba.
\newblock {Artifact for "Partial Order Reduction for Reachability Games"},
  September 2020.

\bibitem[DEF{\etalchar{+}}18]{DEFJJJKLNOPS:FI:18}
A.E. Dalsgaard, S.~Enevoldsen, P.~Fogh, L.S. Jensen, P.G. Jensen, T.S. Jepsen,
  I.~Kaufmann, K.G. Larsen, S.M. Nielsen, M.Chr. Olesen, S.~Pastva, and
  J.~Srba.
\newblock {A Distributed Fixed-Point Algorithm for Extended Dependency Graphs}.
\newblock {\em Fundamenta Informaticae}, 161(4):351--381, 2018.
\newblock IOS Press.

\bibitem[DJJ{\etalchar{+}}12]{david2012tapaal}
A.~David, L.~Jacobsen, M.~Jacobsen, K.Y. Jørgensen, M.H. Møller, and J.~Srba.
\newblock {TAPAAL 2.0: Integrated Development Environment for Timed-Arc Petri
  Nets}.
\newblock In {\em Tools and Algorithms for the Construction and Analysis of
  Systems}, volume 7214 of {\em LNCS}, pages 492--497. Springer Berlin
  Heidelberg, 2012.

\bibitem[DZ04]{10.1007/978-3-540-30468-511}
J.~Dehnert and A.~Zimmermann.
\newblock {Making Workflow Models Sound Using Petri Net Controller Synthesis}.
\newblock In {\em On the Move to Meaningful Internet Systems 2004: CoopIS, DOA,
  and ODBASE}, volume 3290 of {\em LNCS}, pages 139--154. Springer Berlin
  Heidelberg, 2004.

\bibitem[Esp98]{esparza1998decidability}
J.~Esparza.
\newblock {\em {Decidability and Complexity of Petri Net Problems --- An
  Introduction}}, volume 1491 of {\em LNCS}, pages 374--428.
\newblock Springer Berlin Heidelberg, 1998.

\bibitem[GKPP99]{gerth1999partial}
R.~Gerth, R.~Kuiper, D.~Peled, and W.~Penczek.
\newblock {A Partial Order Approach to Branching Time Logic Model Checking}.
\newblock {\em Information and Computation}, 150(2):132--152, 1999.
\newblock Elsevier.

\bibitem[God90]{godefroid1990using}
P.~Godefroid.
\newblock {Using Partial Orders to Improve Automatic Verification Methods}.
\newblock In {\em Computer Aided Verification}, volume 531 of {\em LNCS}, pages
  176--185. Springer Berlin Heidelberg, 1990.

\bibitem[God96]{G:96}
Patrice Godefroid.
\newblock {\em Partial-Order Methods for the Verification of Concurrent
  Systems: An Approach to the State-Explosion Problem}, volume 1032 of {\em
  LNCS}.
\newblock Springer Berlin Heidelberg, Berlin, Heidelberg, 1996.

\bibitem[GW93]{godefroid1993using}
P.~Godefroid and P.~Wolper.
\newblock {Using Partial Orders for the Efficient Verification of Deadlock
  Freedom and Safety Properties}.
\newblock {\em Formal Methods in System Design}, 2(2):149--164, 1993.
\newblock Springer.

\bibitem[HNW98]{huhn1998partial}
M.~Huhn, P.~Niebert, and H.~Wehrheim.
\newblock {Partial Order Reductions for Bisimulation Checking}.
\newblock In {\em Foundations of Software Technology and Theoretical Computer
  Science}, volume 1530 of {\em LNCS}, pages 271--282. Springer Berlin
  Heidelberg, 1998.

\bibitem[JLS16]{jensen2016real}
P.G. Jensen, K.G. Larsen, and J.~Srba.
\newblock {Real-Time Strategy Synthesis for Timed-Arc Petri Net Games via
  Discretization}.
\newblock In {\em Model Checking Software}, volume 9641 of {\em 10580}, pages
  129--146. Springer International Publishing, 2016.

\bibitem[JLS17]{JLS:ICTAC:17}
P.~G. Jensen, K.~G. Larsen, and J.~Srba.
\newblock {PTrie: Data Structure for Compressing and Storing Sets via Prefix
  Sharing}.
\newblock In {\em Proceedings of the 14th International Colloquium on
  Theoretical Aspects of Computing ({ICTAC}'17)}, volume 10580 of {\em LNCS},
  pages 248--265. Springer Berlin Heidelberg, 2017.

\bibitem[JLS18]{jensen2018discrete}
P.G. Jensen, K.G. Larsen, and J.~Srba.
\newblock {Discrete and Continuous Strategies for Timed-Arc Petri Net Games}.
\newblock {\em International Journal on Software Tools for Technology
  Transfer}, 20(5):529--546, 2018.
\newblock Springer Berlin Heidelberg.

\bibitem[JNOS16]{jensen2016tapaal}
J.F Jensen, T.~Nielsen, L.K. Oestergaard, and J.~Srba.
\newblock {TAPAAL and Reachability Analysis of P/T Nets}.
\newblock In {\em Transactions on Petri Nets and Other Models of Concurrency
  XI}, volume 9930 of {\em LNCS}, pages 307--318. Springer Berlin Heidelberg,
  2016.

\bibitem[JPDM18]{jamroga2018towards}
W.~Jamroga, W.~Penczek, P.~Dembiński, and A.~Mazurkiewicz.
\newblock {Towards Partial Order Reductions for Strategic Ability}.
\newblock In {\em Proceedings of the 17th International Conference on
  Autonomous Agents and MultiAgent Systems}, AAMAS 2018, pages 156--165. ACM,
  2018.

\bibitem[KHV16]{kasting2016synthesis}
P.~Kasting, M.R. Hansen, and S.~Vester.
\newblock {Synthesis of Railway-Signaling Plans using Reachability Games}.
\newblock In {\em Proceedings of the 28th Symposium on the Implementation and
  Application of Functional Programming Languages}, IFL 2016, pages 9:1--9:13.
  ACM, 2016.

\bibitem[LLW12]{lehmann2012stubborn}
A.~Lehmann, N.~Lohmann, and K.~Wolf.
\newblock {Stubborn Sets for Simple Linear Time Properties}.
\newblock In {\em Application and Theory of Petri Nets}, volume 7347 of {\em
  LNCS}, pages 228--247. Springer-Verlag, 2012.

\bibitem[LW14]{laarman2014real}
A.~Laarman and A.~Wijs.
\newblock {Partial-Order Reduction for Multi-core LTL Model Checking}.
\newblock In {\em Hardware and Software: Verification and Testing}, volume 8855
  of {\em LNCS}, pages 267--283. Springer Berlin Heidelberg, 2014.

\bibitem[LZ04]{LZ:04}
Z.W. Li and M.C. Zhou.
\newblock Elementary siphons of petri nets and their application to deadlock
  prevention in flexible manufacturing systems.
\newblock {\em IEEE Transactions on Systems, Man, and Cybernetics - Part A:
  Systems and Humans}, 34(1):38--51, 2004.

\bibitem[NWW20]{neele2020partial}
T.~Neele, T.A.C. Willemse, and W.~Wesselink.
\newblock {Partial-Order Reduction for Parity Games with an Application on
  Parameterised Boolean Equation Systems}.
\newblock In {\em Tools and Algorithms for the Construction and Analysis of
  Systems"}, volume 12079 of {\em LNCS}, pages 307--324. Springer Berlin
  Heidelberg, 2020.

\bibitem[Pel93]{peled1993stubborn}
D.~Peled.
\newblock {All From One, One for All: On Model Checking Using Representatives}.
\newblock In {\em Computer Aided Verification}, volume 697 of {\em LNCS}, pages
  409--423. Springer Berlin Heidelberg, 1993.

\bibitem[Pel96]{peled1996combining}
D.~Peled.
\newblock {Combining Partial Order Reductions With On-The-Fly Model-Checking}.
\newblock {\em Formal Methods in System Design}, 8(1):39--64, 1996.
\newblock Springer.

\bibitem[Pel98]{peled1998ten}
D.~Peled.
\newblock {Ten Years of Partial Order Reduction}.
\newblock In {\em Computer Aided Verification}, volume 1427 of {\em LNCS},
  pages 17--28. Springer Berlin Heidelberg, 1998.

\bibitem[QKCC13]{QUINTANILLA2013342}
F.G. Quintanilla, S.~Kubler, O.~Cardin, and P.~Castagna.
\newblock {Product Specification in a Service-Oriented Holonic Manufacturing
  System using Petri-Nets}.
\newblock {\em IFAC Proceedings Volumes}, 46(7):342--347, 2013.
\newblock Elsevier.

\bibitem[RS97]{RS:CONCUR:97}
Y.S. Ramakrishna and S.A. Smolka.
\newblock {Partial-Order Reduction in the Weak Modal Mu-Calculus}.
\newblock In Antoni Mazurkiewicz and J{\'o}zef Winkowski, editors, {\em
  International Conference on Concurrency Theory}, volume 1243 of {\em LNCS},
  pages 5--24. Springer-Verlag, 1997.

\bibitem[Tag08]{T:CIMCA:08}
R.~Tagiew.
\newblock {Multi-Agent Petri-Games}.
\newblock In {\em International Conference on Computational Intelligence for
  Modeling Control Automation}, volume 10981 of {\em LNCS}, pages 130--135.
  {IEEE} Computer Society, 2008.

\bibitem[Val91]{V:90}
A.~Valmari.
\newblock {Stubborn Sets for Reduced State Space Generation}.
\newblock In Grzegorz Rozenberg, editor, {\em Advances in Petri Nets 1990},
  volume 483 of {\em LNCS}, pages 491--515. Springer, 1991.

\bibitem[Val92]{valmari1992attack}
A.~Valmari.
\newblock {A Stubborn Attack on State Explosion}.
\newblock {\em Formal Methods in System Design}, 1(4):297--322, 1992.
\newblock Springer Berlin Heidelberg.

\bibitem[Val93]{valmari1993partial}
A.~Valmari.
\newblock {On-The-Fly Verification with Stubborn Sets}.
\newblock In {\em Computer Aided Verification}, volume 697 of {\em LNCS}, pages
  397--408. Springer Berlin Heidelberg, 1993.

\bibitem[Val97]{valmari1997set}
A.~Valmari.
\newblock {Stubborn Set Methods for Process Algebras}.
\newblock In {\em Proceedings of the DIMACS Workshop on Partial Order Methods
  in Verification}, POMIV'96, page 213–231. Association for Computing
  Machinery, 1997.

\bibitem[VH17]{valmari2017stubborn}
A.~Valmari and H.~Hansen.
\newblock {Stubborn Set Intuition Explained}.
\newblock In {\em Transactions on Petri Nets and Other Models of Concurrency
  XII}, volume 10470 of {\em LNCS}, pages 140--165. Springer Berlin Heidelberg,
  2017.

\bibitem[WW96]{561357}
B.~Willems and P.~Wolper.
\newblock {Partial-Order Methods for Model Checking: From Linear Time to
  Branching Time}.
\newblock In {\em Proceedings 11th Annual IEEE Symposium on Logic in Computer
  Science}, pages 294--303, 1996.

\end{thebibliography}
